\pdfoutput=1
\documentclass[a4paper,UKenglish]{article}
 
\usepackage{microtype}
\usepackage{amsmath}
\usepackage{amsthm}
\usepackage{amssymb}
\usepackage[ruled,vlined,linesnumbered]{algorithm2e}
\usepackage{tikz}
\usepackage{enumitem}
\usepackage{wrapfig}
\usepackage{cutwin}
\usepackage{float}
\usepackage{subfig}
\usepackage{fullpage}
\usepackage{graphicx,import}
\usepackage[hidelinks]{hyperref}

\usepackage[T1]{fontenc}
\usepackage[utf8]{inputenc}
\usepackage[english]{babel}
\usepackage{lmodern}

\newcommand{\CH}{\mathcal{CH}}

\newcommand{\ee}{\varepsilon}
\newcommand{\NN}{\mathbb N}
\newcommand{\ZZ}{\mathbb Z}
\newcommand{\RR}{\mathbb R}
\newcommand{\man}{m}
\newcommand{\manStrat}{M}
\newcommand{\goal}{g}
\newcommand{\lion}{l}
\newcommand{\deltaT}{\Delta}
\newcommand{\rhoD}{\rho}
\newcommand{\rD}{r}
\DeclareMathOperator{\dist}{dist}
\newcommand{\mydef}{\; :=\;}
\newcommand{\dnear}{d_{\textnormal{near}}}
\newcommand{\dfar}{d_{\textnormal{far}}}
\newcommand{\lnear}{l_{\textnormal{near}}}
\newcommand{\lfar}{l_{\textnormal{far}}}

\usepackage{amsthm}
\newtheorem{theorem}{Theorem}
\newtheorem{claim}[theorem]{Claim}
\newtheorem{invariant}[theorem]{Invariant}
\newtheorem{lemma}[theorem]{Lemma}

\title{Best Laid Plans of Lions and Men}

\author{Mikkel Abrahamsen\thanks{Research partly supported by Mikkel Thorup's
    Advanced Grant DFF-0602-02499B from the Danish Council for Independent Research
    under the Sapere Aude research career programme.}
\and Jacob Holm \and 
Eva Rotenberg \and
Christian Wulff-Nilsen\\
\small
Department of Computer Science, University of Copenhagen, Denmark\\
  \texttt{\{miab, jaho, roden, koolooz\}@di.ku.dk}}

\makeatletter
\let\copyrightline\@empty
\makeatother

\theoremstyle{theorem}
\newtheorem{observation}[theorem]{Observation}

\date{} 

\begin{document}

\maketitle

\begin{abstract}
We study the following question dating back to J.E.~Littlewood~(1885--1977):
Can two lions catch a man in a bounded area with rectifiable lakes?
The lions and the man are all assumed to be points moving with at most unit speed.
That the lakes are rectifiable means that their boundaries are finitely
long.
This requirement is necessary 
to avoid pathological examples where the man survives forever because any path
to the lions is infinitely long.

We show that three lions have a winning strategy against a man in a bounded region with finitely many rectifiable lakes. 
This is "tight" in the sense that there exists a region $R$ in the plane where the man has a strategy to
survive forever. We give a rigorous description of such a region $R$; a polygonal region with holes whose exterior and interior
boundaries are pairwise disjoint, simple polygons. 

Finally, we consider the following game played on the entire plane instead of a compact region:
There is any finite number of unit speed lions and one fast man who can
run with speed $1+\ee$ for some value $\ee>0$. Can the man always survive?
We answer the question in the affirmative for any $\ee>0$.
By letting the number of lions tend to infinity, we furthermore show that the man can survive against any countably infinite set of lions.
\end{abstract}

\paragraph{Note:}
The part about the fast man has been published independently~\cite{abrahamsen2020escaping}.
Since the publications of (parts of) this paper~\cite{abrahamsen2017best,abrahamsen2020escaping}, it has come to our attention that two of the main results were already known.
Bhadauria et al.~\cite{bhadauria2012capturing} described a similar polygonal region where the man can survive forever against two lions.
More details are given below.
Chernous'ko~\cite{CHERNOUSKO1976problem} proved already in 1976 that a fast man can escape any finite number of lions.
It seems that the same technique can be used to escape from an infinite set.
For work on a related game and more references, see the paper by Ibragimov et al.~\cite{ibragimov2012evasion}.

\section{Introduction}

`A lion and a man in a closed circular arena have equal maximum speeds.
What tactics should the lion employ to be sure of his meal?'\footnote{The
curve of pursuit ($L$ running always straight at $M$) takes infinite time,
so the wording has its point.} These words (including the footnote) introduce the
now famous lion and man problem, invented by R.~Rado in the
late thirties, in Littlewood's
Miscellany~\cite{littlewood1986littlewood}. It was for a long time believed that
in order to avoid the lion,
it was optimal for the man to run on the boundary of the arena. A simple
argument then shows that the lion
could always catch the man by staying on the radius $OM$ defined by the
man while approaching him as much as possible.
However, A.S.~Besicovitch
proved in 1952 that the man has a very simple strategy
(following which he will approach but not reach the boundary)
that enables him to avoid capture forever no matter
what the lion does. See~\cite{littlewood1986littlewood} for details.

Throughout this paper, all men, lions, and other animals are assumed to be
points.
One can prove that
two lions are enough to catch the man in a circular arena,
and Croft~\cite{croft1964lion}
proves that in general a necessary and sufficient number of 
birds
to catch a
fly inside an $n$-dimensional spherical cage is just
$n$ (again, we assume that the fly and the birds have equal maximum speeds).

A well-known related discrete game is the \emph{cop and robber game}:
Let $G$ be a finite connected undirected graph. Two players called cop $C$ and
robber $R$ play a game on $G$ according to the following rules: First $C$ and then $R$
occupy some vertex of $G$. After that they move alternately along edges of $G$. The
cop $C$ wins if at some point in time $C$ and $R$ are on the same vertex.
If the
robber $R$ can prevent this situation forever, then $R$ wins. The robber
has a winning strategy on many graphs, including all cycles of length at least $4$.
Therefore,
the cop player $C$ can be given a better chance by allowing him,
say, $k$ cops $C_1,\ldots,C_k$.
At every turn $C$ moves any non-empty subset of $\{C_1,\ldots,C_k\}$.
Now, the \emph{cop-number of $G$} is the minimal number of cops needed for
$C$ to win. Aigner and Fromme~\cite{aigner1984game}
observe that the cop-number of the dodecahedron graph is at least
$3$, since if there are only $2$ cops,
the robber can always move to a vertex
not occupied by a cop and not in the neighbourhood of any. Furthermore,
they prove that the cop-number of any
planar graph is at most $3$. Thus, the cop-number of the dodecahedron is
exactly $3$.

Returning to the lion and man game, 
Bollob{\'a}s~\cite{bollobas2006art} writes that the following open problem
was already mentioned by J.E. Littlewood~(1885--1977):
Can two lions catch a man in a bounded (planar) area with rectifiable lakes?
An informal definition of a rectifiable curve
is that it has finite length.
We require that the boundaries of the lakes and the exterior boundary
are all rectifiable curves to avoid pathological examples
where the man survives forever, e.g.~because any path to the lions is infinitely long.
Bollob{\'a}s mentions the same problem in a comment in his edition of Littlewood's Miscellany~\cite{littlewood1986littlewood} and
in~\cite{bollobas2011lion}. The problem is also stated by Fokkink et al.~\cite{fokkink2013open}.
Berarducci and Intrigila~\cite{berarducci1993cop} prove that the man can
survive forever (for some initial positions of the man and lions)
if the area is a planar embedding of the dodecahedron graph
where each edge is a curve with the same
length, say length $1$. The proof is essentially the same as the proof by Aigner and
Fromme~\cite{aigner1984game} that the cop-number of the dodecahedron
is at least $3$: When the man is standing at a vertex, there will always be a neighbouring vertex with distance more than $1$ to the nearest lion. It is thus safe for the man to run to that vertex. This, however, is a one-dimensional example.
Berarducci and Intrigila raise the question whether it is possible to replace the one-dimensional edges
by two-dimensional thin lines.

The question by Berarducci and Intrigila is answered in the affirmative by Bhadauria et al.~\cite{bhadauria2012capturing}:
They present the game in the discrete time model where the players take turns.
When the man has his turn, he moves to a point within geodesic distance at most $1$.
Likewise, when the lions have their turn, each of them moves to a point within distance at most $1$.
It is easy to see that it is an advantage for the lions to play in the discrete time model as compared to the continuous model that we use in this paper, since, if the lions have a winning strategy in the continuous time model, the lions can simulate that strategy in the discrete model.
Therefore, the result that three lions are sometimes necessary in the continuous model follows from the work of Bhadauria et al.

In Section~\ref{sec:threeLions}, we show our first main result; that three lions are always enough to catch the man in a compact region with a finite number of rectifiable lakes.
Bhadauria et al.~\cite{bhadauria2012capturing} showed that this is also the case in the discrete time model in a polygonal region with polygonal lakes.
Our result does not follow from theirs since we do not assume the region to be polygonal. Furthermore, more lions are sometimes needed in the continuous time model as compared to the discrete time model.
For instance, when a man is playing against one lion in a closed disk of radius $r>0$, the lion has a strategy to win after $O(r^2)$ turns in the discrete time model~\cite{chung2011search}, whereas the man can always win in the continuous time model (this is the original game introduced by Rado).

In Section~\ref{unitSpeedManSec}, we include a description of a construction similar to Bhadauria et al.~\cite{bhadauria2012capturing} which we found independently\footnote{The result by Bhadauria et al.~\cite{bhadauria2012capturing} came to our attention after the preliminary version of the present paper appeared~\cite{abrahamsen2017best}.}. This is included for the sake of completeness, although our sole contribution is that we give a more detailed description than the one given by Bhadauria et al.
In particular, we provide a complete drawing of a region where the man can win against two lions, see Figure~\ref{fig:region}.

Rado and Rado~\cite{rado1974more} and
Jankovi{\'c}~\cite{jankovic1978about}
consider the problem where there are many lions
and one man, but where the game is played in the entire unbounded plane.
They prove that the lions can catch the man
if and only if the man starts in the interior of the convex hull of
the lions. Inspired by that problem,
we ask the following question: What if the lions have maximum speed
$1$ and the man has maximum speed $1+\ee$ for some $\ee>0$? As our second main result, we prove that
for any $\ee >0$, and any finite number of lions, such a fast man can survive forever provided that
he does not start at the same point as one of the lions.
We explain a strategy 
in Section~\ref{fastManSec}.
In Section~\ref{infLionsSec}, we show how to extend the argument to obtain a winning strategy against any countably infinite set of lions.
We find this result somewhat surprising.
Indeed, it is difficult to imagine how that man proceeds if, say, he starts in the point $(\sqrt 2,0)$ and there are lions at all points with two rational coordinates.

Other variants of the game with a faster man have been studied previously.
Flynn~\cite{flynn1973lion,flynn1974lion} and
Lewin~\cite{lewin1986lion} study the problem where there is one lion and one fast man
in a circular arena. The lion tries to get as close to the man as possible and the
man tries to keep the distance as large as possible.
Variants of the cop and robber game where the robber is faster than the cops
have also been studied. See for instance~\cite{alon2015chasing,fomin2010pursuing}.

\subsection{Definitions}\label{definitions}

We follow the conventions of Bollob{\'a}s et al.~\cite{bollobas2012lion}.
Let $R\subseteq\RR^2$ be a region in the plane on which the lion and man
game is to be played, and assume that the lion starts at point $l_0$ and the man
at point $m_0$. We define a \emph{man path} as a function
$m\colon [0, \infty)\longrightarrow R$ satisfying $m(0)=m_0$ and
the Lipschitz condition
$\lVert m(s)-m(t)\rVert\leq V\cdot \lvert s-t\rvert$, where $V$ is the speed of the man.
In our case, we either have $V=1$ or, in the case of a fast man,
$V=1+\ee$ for some small
$\ee>0$. Note that it follows from the Lipschitz condition
that any man path is continuous.
A \emph{lion path} $l$ is defined similarly, but the lions we consider always run
with at most unit speed.
Let $\mathcal L$ be the set of all lion paths and $\mathcal M$ be the set of all
man paths. Then a \emph{strategy} for the man
is a function $M\colon \mathcal L\longrightarrow\mathcal M$ such that
if $l,l'\in\mathcal L$ agree on $[0,t]$, then $M(l)$ and $M(l')$ also agree on
$[0,t]$. This last condition is a formal way to describe that the man's position
$M(l)(t)$, when he follows strategy $M$, depends only on
the position of the lion at points in time
before and including time $t$, i.e., he is not allowed to act based on the lion's future
movements.
(By the continuity of any man path, the man's position at time $t$ is in fact
determined by the lion's position at all times strictly before time $t$.)
A strategy $M$ for the man is \emph{winning} if for any $l\in\mathcal L$ and
any $t\in[0,\infty)$, it holds that $M(l)(t)\neq l(t)$.
Similarly, a strategy for the lion $L\colon\mathcal M\longrightarrow\mathcal L$
is winning if for any $m\in\mathcal M$, it holds that $L(m)(t)=m(t)$
for some $t\in[0,\infty)$.
These definitions are extended to games with more than one lion in the natural way.

It might seem unfair that the lion is not allowed to react on the man's movements
when we evaluate whether a strategy $M$ for the man is winning. However,
we can give the lion full information about $M$ and
allow it to choose its path $l$ depending on $M$ \emph{prior}
to the start of the game. If $M$ is a winning strategy, the man can also
survive the lion running along $l$.

We call a man strategy
$M$ \emph{locally finite} if it satisfies the following property:
if $l$ and $l'$ are any two
lion paths that agree on $[0, t]$ for some $t$ then the corresponding man paths
$M(l)$ and $M(l')$ agree on $[0, t+\delta]$ for some
$\delta > 0$ (we allow that $\delta$ depends on $l\vert_{[0,t]}$).
Thus, informally, the man commits to doing something for some positive
amount of time dependent only on the situation so far.
Bollob{\'a}s et al.~\cite{bollobas2012lion} prove that if the man has a locally finite
winning strategy, then the lion does not have any winning strategy.
The argument easily extends to
games with multiple lions.
At first sight, it might sound absurd to even consider the possibility that the lion
has a winning strategy when the man also does. However, it does not follow
from the definition that the existence of a winning strategy for the man
implies that the lion does not also have a winning strategy.
See the paper by Bollob{\'a}s et al.~\cite{bollobas2012lion} for a detailed discussion
of this (including descriptions of natural variants of the lion and man game where
both players have winning strategies). In each of the problems we describe, the winning
strategy of the man is locally finite, so it follows that the lions do not
have winning strategies. In fact, the strategies we describe satisfy the much stronger
condition that they are \emph{equitemporal}, i.e., there is a $\Delta>0$ such that the man at any point in time $i\cdot\Delta$, for $i=0,1,\ldots$,
decides where he wants to run until time $(i+1)\cdot\Delta$.

\newcommand{\outerb}{^\ast}
\newcommand{\dd}{d}

\section{Sufficiency of Three Lions}\label{sec:threeLions}

Consider a region $R\subset\RR^2$ that is closed, bounded, connected, has a rectifiable boundary, and finitely many lakes.
A \emph{lake} of $R$ is a bounded connected component of $\RR^2\setminus R$.
We denote by $\partial R$ the boundary of $R$ and by $\partial R^*$ the \emph{exterior} boundary of $R$, i.e., the boundary of the unbounded connected component of $\RR^2\setminus R$.
Assume for simplicity that the exterior boundary and the boundary of each lake is a simple closed curve and that the boundaries of the lakes and the exterior boundary of $R$ are all pairwise disjoint.
We say that such a region $R$ is \emph{feasible}.

In this section, we show the following result.
\begin{theorem}\label{threeLionsWin}
Let $R$ be a feasible region.
Then three lions can always catch one man in $R$ within finite time.
\end{theorem}

In the following, we let $R$ denote the feasible region from Theorem~\ref{threeLionsWin}.
The idea behind the proof of the theorem is that the lions first use one strategy to restrict the man to a subset of $R$ with no lakes and then use another strategy to catch him.

\subsection{Restricting the man to a region with no lakes}

Let $P$ be a path in $R$.
We say that a lion $\lion$ \emph{guards} $P$ if $\lion$ moves in such a way that it will catch the man as soon as he steps on a point on $P$.


\begin{figure}
\centering
\includegraphics{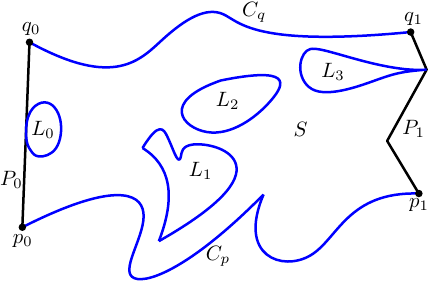}\hspace{0.1\textwidth}\includegraphics{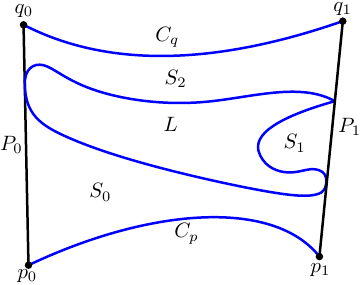}
\caption{Left: A strip $S$. $L_0$ is a boundary lake of $P_0$, $L_1$ and $L_2$ are inner lakes of $S$, and $L_3$ is a boundary lake of $P_1$.
In this strip, both $P_0$ and $P_1$ are shortest paths, so $S$ trivially satisfies the guarding condition.
Right: A strip $S=S_0\cup S_1\cup S_2$ which is not minimal, since the man will be restricted to either $S_0$, $S_1$, or $S_2$.
The strip $S_1$ is degenerate in the sense that one of the paths forming the exterior boundary is degenerate.}
\label{fig:stripNew}
\end{figure}

Define a \emph{strip} $S\mydef S(P_0,P_1,C_p,C_q)$ (see Figure~\ref{fig:stripNew} (left)) to be a closed subset of $R$ with an exterior boundary consisting of four curves $P_0,P_1,C_p,C_q$ such that
\begin{itemize}
\item
$S$ is a maximal subset of $R$ in the sense that $S$ contains every subset of $R$ with exterior boundary $P_0\cup P_1\cup C_p\cup C_q$,

\item
$P_0$ and $P_1$ are paths of finite length and have endpoints $p_0,q_0$ and $p_1,q_1$, respectively,

\item
$C_p,C_q$ are on the boundary $\partial R$ and have endpoints
$p_0,p_1$ and $q_0,q_1$, respectively,

\item
$P_0$ and $P_1$ separate $S$ from $R\setminus S$ in the sense that any path in $R$ from a point in $S$ to a point in $R\setminus S$ intersects $P_0$ or $P_1$, and

\item
the following condition, denoted as the \emph{guarding condition}, holds for a value $u\in\{0,1\}$:
$P_u$ is a shortest path from $p_u$ to $q_u$ among all paths in $S$, and
$P_{1-u}$ is a shortest path among all paths $P'$ in $S$ from $p_{1-u}$ to $q_{1-u}$ satisfying that the region enclosed by $P_u\cup P'\cup C_p\cup C_q$ contains all the lakes of $S$ touching $P_u$.
\end{itemize}
We allow any of $P_0,P_1,C_0,C_1$ to degenerate to a single point.
Note that a strip is a feasible region.
We orient $P_0$ and $P_1$ from $C_p$ to $C_q$ so that we have a notion of up and down.

An \emph{inner lake} of a strip $S$ is a lake of $S$ with a boundary disjoint from both $P_0$ and $P_1$.
A \emph{boundary lake} is a lake of $S$ that is not an inner lake, i.e., a lake with a boundary intersecting $P_0$ or $P_1$.
If $P_u$, $u\in\{0,1\}$, touches the boundary of a lake $L$, we say that $P_u$ \emph{has} the boundary lake $L$.
Note that since we assume that the boundaries of the lakes and the exterior boundary of $R$ are all pairwise disjoint, it follows that no lake has a boundary that only touches endpoints of $P_0,P_1$.

Suppose that lions $\lion_0$ and $\lion_1$ guard $P_0$ and $P_1$, respectively, and that the man is contained in the strip $S$.
We say that $S$ is \emph{minimal} if no strip $S'$ which is a proper subset of $S$ exists such that $\partial S'\subset \partial S$ and $S'$ contains the man.
To put it in another way, $S$ is minimal if the man is not restricted to any strip $S'$ properly contained in $S$.
Note that a lake $L$ of a minimal strip $S$ cannot be a boundary lake of both $P_0$ and $P_1$, since otherwise, $L$ would restrict the man to be either in a smaller strip above or below $L$.

Consider a strip $S\mydef S(P_0,P_1,C_p,C_q)$ with one or more lakes.
Assume that the man is in $S$ and that two lions $\lion_0$ and $\lion_1$ guard $P_0$ and $P_1$, respectively, so that the man cannot escape from $S$.
We now describe how the three lions proceed to catch the man in $S$.
\begin{itemize}
\item $S$ is not minimal, i.e., there is a strip $S'$ which is a proper subset of $S$ such that $\partial S'\subset \partial S$ and $S'$ contains the man.
See Figure~\ref{fig:stripNew} (right).
In this case, the lions perform a \emph{shrinking move}, which is just a formal concept and does not require the lions to do anything actively.
Let the portions of $P_0$ and $P_1$ appearing on the boundary of $S'$ be $P'_0$ and $P'_1$, respectively (one of $P'_0,P'_1$ might be empty).
By definition, $\lion_0$ is guarding $P'_0$ and $\lion_1$ is guarding $P'_1$.
Therefore, the lions $\lion_0,\lion_1$ are already restricting the man to the strip $S'$ which is smaller than $S$.

\begin{figure}
\centering
\includegraphics{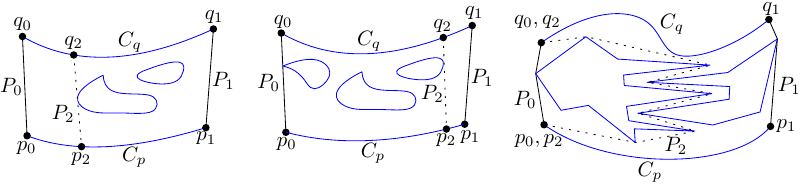}
\caption{A cutting move of each of the three types.}
\label{fig:cuttingMoves}
\end{figure}

\item $S$ is minimal.
In this case, the lions perform a \emph{cutting move}, to be defined in the following.
Orient $P_0$ from $p_0$ to $q_0$ and $P_1$ from $p_1$ to $q_1$ and assume that $S$ is to the right of $P_0$ and to the left of $P_1$.
The strategy of the third lion $\lion_2$ is to guard a path $P_2$ from $p_2\in C_p$ to $q_2\in C_q$ defined by cases as follows (see Figure~\ref{fig:cuttingMoves}).
\begin{enumerate}
\item\label{strat:1}
Neither $P_0$ nor $P_1$ has a boundary lake of $S$.
We choose $p_2$ and $q_2$ such that a shortest path $P_2$ in $S$ from $p_2$ to $q_2$ touches the boundary of a lake of $S$.

\item\label{strat:2}
For $u\in\{0,1\}$ the path $P_u$ has a boundary lake of $S$, but $P_{1-u}$ does not have any.
We choose $p_2$ and $q_2$ such that a shortest path $P_2$ in $S$ from $p_2$ to $q_2$ is on the same side of all lakes as $P_{1-u}$ and touches one or more lakes.

\item\label{strat:3}
$P_0$ and $P_1$ both have boundary lakes of $S$.
Recall that the boundary lakes are uniquely assigned to either $P_0$ or $P_1$.
In this case, we choose $P_2$ to be a shortest path in $S$ from $p_0$ to $q_0$ such that all boundary lakes of $P_0$ are to the left of $P_2$ and all boundary lakes of $P_1$ are to the right of $P_2$.
In other words, $P_2$ is a shortest path from $p_0$ to $q_0$ that separates the boundary lakes of $P_0$ from the ones of $P_1$ (or, in short, a shortest path from $p_0$ to $q_0$ that \emph{separates lakes}).
\end{enumerate}
\end{itemize}

We first verify the existence of the path $P_2$ as defined in each of the above cases of the cutting move.
\begin{enumerate}
\item[\ref{strat:1}.]
It follows from the guarding condition of $S$ that $P_0$ and $P_1$ are both shortest paths in $S$.
Consider a shortest path $\Pi(p,q)$ from $p\in C_q$ to $q\in C_q$ as we slide $p$ along $C_p$ from $p_0$ to $p_1$ and $q$ along $C_q$ from $q_0$ to $q_1$.
When $(p,q)=(p_0,q_0)$, all lakes of $S$ are to the right of $\Pi(p,q)$, whereas when $(p,q)=(p_1,q_1)$, all lakes are to the left of $\Pi(p,q)$.
Therefore, $\Pi(p,q)$ must at some point jump over a lake.
As this happens, $\Pi(p,q)$ must touch the boundary of a lake.

\item[\ref{strat:2}.]
As in the previous case, it follows from the guarding condition that $P_0$ and $P_1$ are both shortest paths in $S$.
Assume without loss of generality that $P_0$ has no boundary lakes.
We now perform a sliding argument similar to the one in case~\ref{strat:1}.

\item[\ref{strat:3}.]
The existence of a path $P_2$ as described is obvious in this case.
\end{enumerate}

\begin{lemma}\label{lem:fewer_lakes}
Consider a strip $S\mydef S(P_0,P_1,C_p,C_q)$ with one or more lakes.
Assume that $S$ contains the man, $S$ is minimal, and that two lions $\lion_0$ and $\lion_1$ guard $P_0$ and $P_1$, respectively, so that the man cannot escape from $S$.
Then the third lion $\lion_2$ can after finite time guard a path $P_2$ in $S$ from $p_2\in C_p$ to $q_2\in C_q$ as specified in the above strategy for the cutting move, and $P_2$ divides $S$ into two strips $S_l,S_r$ to the left and right side of $P_2$, respectively, each with fewer lakes than $S$ or with the same lakes where some that are inner lakes in $S$ have now become boundary lakes.
\end{lemma}

\begin{proof}
We first prove that $\lion_2$ can after finite time guard the path $P_2$ specified by the strategy.
Let the endpoints of $P_2$ be $p_2\in C_p$ and $q_2\in C_q$.
Note that in each of the three cases of the strategy, it holds that $P_2$ is a shortest path from $p_2$ to $q_2$ that separates lakes (even though there might be no lakes touching $P_0$ or $P_1$).

For any point $r\in S$, let $\Pi_r$ be a shortest path in $S$ from $p_2$ to $q_2$ passing through $r$ that separates lakes.
Let $d_p(r)\mydef \|\Pi_r[p_2,r]\|$, where $\Pi_r[p_2,r]$ is the portion of $\Pi_r$ from $p_2$ to $r$ and $\|\Pi\|$, for a path $\Pi$, is the length of $\Pi$.

The strategy of the lion $\lion_2$ is to stand, at any time $t$, at the point $\pi(t)$ on $P_2$ such that
$$\frac{d_p(\pi(t))}{\|P_2\|}=\frac{d_p(m(t))}{\|\Pi_{\man(t)}\|}.$$
We first verify that the point $\pi(t)$ moves with at most unit speed.
To this end, let two points in time $t_0,t_1$ be given.
We have the following bounds.
\begin{align*}
\|\pi(t_0)-\pi(t_1)\|
&\leq\|P_2[\pi(t_0),\pi(t_1)]\|\\
&=|d_p(\pi(t_0))-d_p(\pi(t_1))|\\
&=\|P_2\|\cdot \left|\frac{d_p(\man(t_0))}{\|\Pi_{\man(t_0)}\|}-\frac{d_p(\man(t_1))}{\|\Pi_{\man(t_1)}\|}\right|\\
&\leq \|P_2\|\cdot \frac{|d_p(\man(t_0))-d_p(\man(t_1))|}{\min\{\|\Pi_{\man(t_0)}\|,\|\Pi_{\man(t_1)}\|\}}\\
&\leq |d_p(\man(t_0))-d_p(\man(t_1))|.
\end{align*}

To finish the argument that $\pi(t)$ moves with at most unit speed, we need to verify that $|d_p(\man(t_0))-d_p(\man(t_1))|\leq |t_0-t_1|$.
To this end, note that the path
$$\Pi_{\man(t_0)}[p_2,\man(t_0)]\cup \man([t_0,t_1])\cup\Pi_{\man(t_1)}[\man(t_1),q_2]$$
must separate lakes since otherwise, the path $\man([t_0,t_1])$ crosses $P_0$ or $P_1$ and hence the man would have been caught by $\lion_0$ or $\lion_1$ by assumption.
Therefore,
$$d_p(\man(t_1))\leq d_p(\man(t_0))+\|\man([t_0,t_1])\|\leq d_p(\man(t_0))+|t_0-t_1|.$$
Likewise, the path
$$\Pi_{\man(t_1)}[p_2,\man(t_1)]\cup \man([t_1,t_0])\cup\Pi_{\man(t_0)}[\man(t_0),q_2]$$
separates lakes, and it follows that
$$d_p(\man(t_0))\leq d_p(\man(t_1))+\|\man([t_0,t_1])\|\leq d_p(\man(t_1))+|t_0-t_1|.$$
From these two bounds, the desired inequality follows.

Hence, the point $\pi(t)$ moves with at most unit speed.
In order to get to the point $\pi(t)$ in the first place, the lion $\lion_2$ first runs to $p_2$ and then follows $P_2$ to $q_2$.
It must eventually get to the point $\pi(t)$ and then starts following that point.
To see that the lion guards $P_2$, note that if at some point in time $t$ the man steps on a point in $P_2$, we have $\pi(t)=m(t)$, so the lion catches the man.

It remains to verify that the path $P_2$ separates $S$ into two strips $S_l,S_r$ as stated.
Let $S_l$ be the part of $S$ on or to the left of the path $P_2$ and $S_r$ the part on or to the right of $P_2$.
We claim that $S_l$ and $S_r$ are strips.
Due to symmetry, it suffices to verify that $S_l$ is a strip.
The exterior boundary of $S_l$ is $P_0\cup P_1\cup C_p[p_0,p_2]\cup C_q[q_0,q_2]$.
Recall that $S$ satisfies the guarding condition for either $u=0$ or $u=1$.
We need to verify that in each case, $S_l$ also satisfies the guarding condition.
\begin{itemize}
\item $u=0$. In this case, $P_0$ is a shortest path in $S$ and thus also in $S_l$.
Since $P_2$ is a shortest path in $S$ that separates lakes, it follows that $P_2$ is a shortest path in $S_l$ among all paths in $S_l$ to the right of the lakes touching $P_0$.
Hence, $S_l$ satisfies the guarding condition and it follows that $S_l$ is a strip.

\item $u=1$. In this case, $P_0$ is shortest among all paths in $S$ to the left of the lakes touching $P_1$.
Since $P_2$ separates lakes in $S$, it follows that $P_0$ is a shortest path in $S_l$.
As in the other case, $P_2$ is a shortest path in $S_l$ among all paths in $S_l$ to the right of the lakes touching $P_0$.
Hence, $S_l$ satisfies the guarding condition, so $S_l$ is a strip.
\end{itemize}

In the cases~\ref{strat:1} and~\ref{strat:2} of the strategy for the cutting move, the strip $S_l$ either has no lakes or the same lakes as $S$ where one or more inner lakes of $S$ are boundary lakes of $S_l$.
In the case~\ref{strat:3}, each strip $S_l,S_r$ has fewer lakes than $S$.
This finishes the proof.
\end{proof}

From the very beginning, where the man is only restricted to the entire region $R$, the lions proceed as follows.
Initially, the lions $\lion_0,\lion_1$ run to two arbitrary points on the exterior boundary of $R$.
Then each lion guards a degenerate path of the strip $R$.
We now show inductively how the three lions collaborate in order to restrict the man to a strip containing no lakes.
Suppose that $\lion_0,\lion_1$ guard paths $P_0,P_1$ on the exterior boundary of a minimal strip $S$ containing the man as described in Lemma~\ref{lem:fewer_lakes}.
The lion $\lion_2$ applies Lemma~\ref{lem:fewer_lakes} and guards a path $P_2$ that separates $\lion_0$ and $\lion_1$.
The man is either restricted to $S_l$ or $S_r$.
If the strip is not minimal, the lions make a shrinking move and guard the minimal strip to which the man is restricted.
When the lions guard a minimal strip, there is an idle lion that can again apply Lemma~\ref{lem:fewer_lakes}.

It takes no time to perform the shrinking moves.
In each cutting move, a lake is either completely eliminated, or an inner lake is turned into a boundary lake.
Thus, after $2L$ cutting moves, where $L$ is the number of lakes of $R$, the man is restricted to a strip $S\mydef S(P_0,P_1,C_p,C_q)$ with no lakes.

\subsection{Catching the man in a strip with no lakes}

Suppose that the man is restricted to a strip $S\mydef S(P_1,P_2,C_p,C_q)$ with no lakes where $\lion_0$ guards $P_0$ and $\lion_1$ guards $P_1$ and that $P_0$ has endpoints $p_0\in C_p$ and $q_0\in C_q$, and $P_1$ has endpoints $p_1\in C_p$ and $q_1\in C_q$.
Recall that $\lion_0$ guards $P_0$ by keeping on the point $\pi_0(t)$ as defined in Lemma~\ref{lem:fewer_lakes}.
We now introduce an alternative way of guarding a path in a region without lakes, which we shall make use of in the rest of this section.
Define $\pi'_0(t)$ to be the point on $P_0$ with the smallest geodesic distance in $S$ to $\man(t)$.
Since there are no lakes in $S$ and $P_0$ is a shortest path, it is easy to see that $\pi'_0(t)$ moves with at most unit speed.
Now, $\lion_2$ moves to $\pi'_0(t)$ by following $P_0$ from $p_0$ and then guards $P_0$ by staying on $\pi'_0(t)$.
Thus $\lion_2$ starts guarding $P_0$ and this makes $\lion_0$ idle.
Let $\pi'_1(t)$ be the point on $P_1$ with the smallest geodesic distance in $S$ to $\man(t)$.
$\lion_0$ now runs to $\pi'_1(t)$ along $P_1$ from $p_1$ and thus makes $\lion_1$ idle.
In the following, when we say that a lion guards a shortest path, we mean that it stays on the point on the path closest to the man.

The lion $\lion_1$ is idle and can thus start guarding the shortest path in $S$ from $p_0$ to $q_1$, see Figure~\ref{fig:peninsula} (left).
The man is now restricted a \emph{peninsula} $Q$, which is a subset of $R$ with no lakes such that the boundary of $Q$ consists of two paths which are shortest paths in $Q$ and one portion $C$ of the boundary of $R$, denoted as the \emph{coast} of $Q$.
Define a \emph{pseudo-triangle} $T$ to be a subset of $R$ with no lakes such that the boundary of $T$ consists of three paths which are shortest paths in $T$, see Figure~\ref{fig:peninsula} (middle).
It follows that the paths are concave with respect to $T$.

\begin{figure}
\centering
\includegraphics{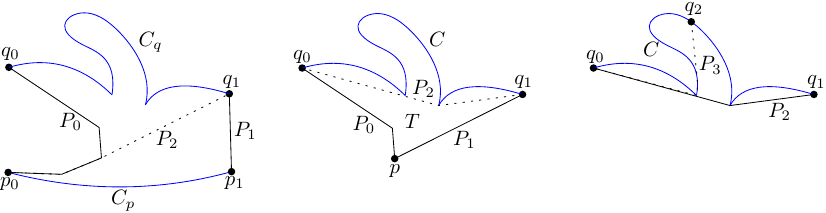}
\caption{The lions' strategy to catch the man in a strip without lakes.
In the middle figure, the region $T$ bounded by $P_0\cup P_1\cup P_2$ is a pseudo-triangle.}
\label{fig:peninsula}
\end{figure}

\begin{lemma}
Suppose that the man is restricted to a peninsula $Q$.
Let the boundary of $Q$ consist of two paths $P_0,P_1$ from the common endpoint $p$ to the points $q_0$ and $q_1$, respectively, and the coast $C$ from $q_0$ to $q_1$.
Suppose that two lions guard $P_0$ and $P_1$.
After $O(\|C\|)$ time, the three lions can restrict the man to either a pseudo-triangle or a peninsula $Q'\subset Q$ such that the coast of $Q'$ is half as long as the coast of $Q$.
\end{lemma}

\begin{proof}
Suppose that $\lion_0$ guards $P_0$ and $\lion_1$ guards $P_1$.
Now, $\lion_2$ starts guarding the shortest path $P_2$ from $q_0$ to $q_1$, see Figure~\ref{fig:peninsula} (middle).
Since $P_2$ is at most as long as $C$, it takes $O(\|C\|)$ time for $\lion_2$ to start guarding $P_2$.
Note that $P_2$ divides $Q$ into a pseudo-triangle and a degenerated peninsula $Q_0$ with the same coast as $Q$.
If the man is contained in the pseudo-triangle, we are done.
Otherwise, if the man is in $Q_0$, the lions $\lion_0$ and $\lion_1$ have both become idle.
Let $q_2$ be the middle point on $C$, i.e., the point such that $\|C[q_0,q_2]\|=\|C[q_2,q_1]\|=\|C\|/2$.
The lion $\lion_0$ now guards the shortest path $P_3$ from $q_0$ to $q_2$, see Figure~\ref{fig:peninsula} (right).
\textbf{}This takes time $O(\|C\|)$.
Note that $P_3$ divides $Q_0$ into two peninsulas, each with a coast of half the length of $C$, so the statement follows.
\end{proof}

Note that if the man keeps being in a peninsula, the coast of the peninsula converges to a single point after $O(\|C\|)$ time, where $C$ is the coast of the first peninsula.
In this case, it thus follows that the lions catch the man in finite time.
It remains to describe what happens if the man is restricted to a pseudo-triangle.

\begin{figure}
\centering
\includegraphics{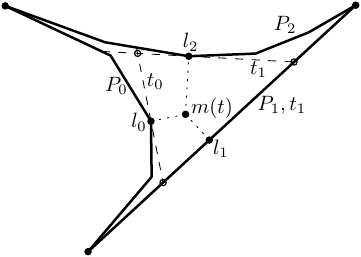}
\caption{The lions' strategy to catch the man in a pseudo-triangle $T$.
The corners of the triangle $T'\subset T$ are drawn as circles.}
\label{fig:pseudotriangle}
\end{figure}

\begin{lemma}
If the man is restricted to a pseudo-triangle $T$ with a boundary consisting of three shortest paths $P_0,P_1,P_2$ guarded by the three lions $\lion_0,\lion_1,\lion_2$, respectively, then they can catch the man inside $T$ after finite time.
\end{lemma}

\begin{proof}
See Figure~\ref{fig:pseudotriangle}.
Each of the paths $P_0,P_1,P_2$ is concave with respect to $T$, and each lion $\lion_i$ is on the point of $P_i$ closest to the man.
Suppose that the paths are oriented following the counterclockwise traversal of the boundary of $T$.
For each $i\in\{0,1,2\}$, consider the tangent $t_i$ to $P_i$ through the lion $\lion_i$ perpendicular to the line segment $\lion_i(t)\man(t)$.
Now, instead of guarding $P_i$, the lion $\lion_i$ starts guarding the part of $t_i$ contained in $T$.
It follows that the man is to the left of $t_i$ for each $i\in\{0,1,2\}$.
Thus, the man is now restricted to a triangle $T'\subseteq T$ in the usual sense.
The three lions now apply the strategy described by Jankovi{\'c}~\cite{jankovic1978about} in order to catch the man.
In short, each lion keeps guarding a line segment parallel to the segment it was guarding in $T'$ by staying at the projection of the man on the segment.
The lion uses its extra speed (if any) to approach the man, i.e., moving the segment that it guards towards the man, thus shrinking the triangle that the lions together guard.
There is always one lion that can approach the man with a speed which is at least some positive constant.
Therefore, the lions catch the man after finite time.
\end{proof}

This finishes the proof of Theorem~\ref{threeLionsWin}.
We conclude the section by mentioning that the lions' strategy can be adapted to a locally finite strategy by a similar technique as used by Bollob{\'a}s et al.~in the appendix of~\cite{bollobas2012lion}.
They study the \emph{porter and student game}, which is played in the square $[-1,1]^2$.
The student starts at $(0,0)$, whereas the four porters start at points $(0,\pm 1)$ and $(\pm 1,0)$.
The porters are restricted to the boundary of the square.
The student wins if she can reach the boundary of the square without being caught as she reaches it.
The obvious strategy for the porters is that each porter stays on the side of the square where he starts and keeps the other coordinate the same as the student's.
That strategy is not locally finite, but Bollob{\'a}s et al.~show how it can be discretized so that it becomes locally finite.
Hence, the student does not have a winning strategy.
A similar discretization can be used by the three lions, and it follows that the man cannot win against three lions.

\section{Necessity of Three Lions}\label{unitSpeedManSec}
In this section, we present a polygonal region $R$ in the plane with $11$ lakes, where the man may successfully escape two lions. Our construction similar to Bhadauria et al.~\cite{bhadauria2012capturing}, and is included for the sake of completeness. 
See Figure~\ref{fig:region} for an illustration of such a region.
The region in Figure~\ref{fig:region} is made with the same techniques as the region $R$ described in this section, but some parameters were tweaked in order to make the details more easily visible when shown within a single page.

The exterior and interior boundaries of $R$ are all pairwise disjoint simple polygons, and a man can survive forever in $R$ against two lions provided that the lions are initially at a sufficient distance.

Consider a planar embedding $\mathcal D$
of the dodecahedron where each edge is
a polygonal curve.
We can obtain that all edges have the same length by prolonging some edges using
a zig-zag pattern.
This embedding
corresponds to an area with $11$ lakes and infinitely thin paths between the lakes,
and as observed by Berarducci and Intrigila~\cite{berarducci1993cop},
the man can survive forever against two lions on such an embedding by
deciding at each vertex which neighbouring vertex to visit next.
First, we explain why it is not straight-forward to obtain the region $R$ from
$\mathcal D$, or, at least, why some natural initial attemps will not work.

We want to ``thicken'' each edge of $\mathcal D$
such that the boundaries of the lakes become disjoint, thus
obtaining a truly two-dimensional region $\mathcal D'$ containing
$\mathcal D$ as a subset.
However, doing so, the point in $\mathcal D'$ corresponding to a vertex of $\mathcal{D}$ does not necessarily lie on the shortest path between its neighbours. We thus cannot simply employ the strategy from $\mathcal{D}$, roughly speaking, because the man must plan in advance which turn to take in the upcoming vertex.

In order to carry out this idea, we first need to describe a
winning strategy of the man on the dodecahedron graph with the special
property that he does not make his decisions at the vertices.
Let $\mathcal G$ be a planar
embedding of the dodecahedron where all edges have length $4$.
The distance between two points in $\mathcal G$ is the length of a shortest
path between the points.
Let the \emph{quarters} denote the points on the edges of $\mathcal G$
at distance $1$ to the closest vertex.
Consider a quarter $x$ on the edge $ab$ of $\mathcal G$.
For a point $p\in\mathcal G$, $p\neq x$,
let $d_a(x,p)$ be the length of a shortest \emph{simple} path in
$\mathcal G$ from $x$ to $p$ that initially follows the edge $\{a,b\}$ in direction
towards $a$.
Let $d_b(x,p)$ be defined similarly.

When the man is at a quarter
$x$ with distance $1$ to the vertex $a$ and $3$ to the vertex $b$, we let
$\dnear$ denote the distance from $x$ to the closest 
lion with respect to $d_a$, and let $\dfar$ denote the distance from $x$ to the closest lion with respect to $d_b$.
To avoid confusion, we write them as $\dnear(t)$ and $\dfar(t)$ when $x$ is the position of the man at the time $t$.

\begin{figure}[!h]
    \centering
    \includegraphics[width=\textwidth]{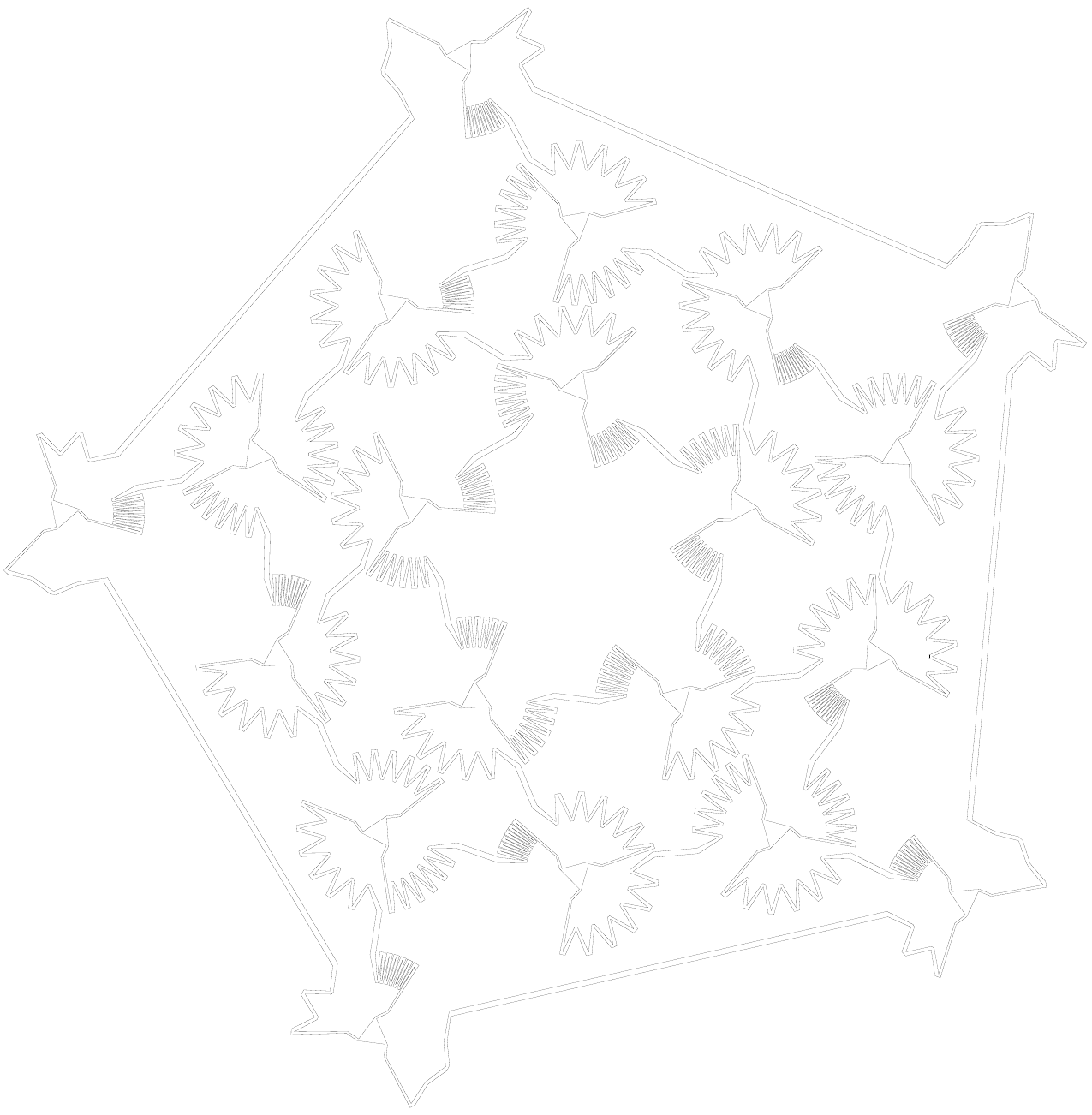}
    \vspace{-70pt}	        
    \caption{A region with $11$ lakes in which the man has a winning strategy against two lions. The man and the lions are restricted to the triangular rooms and the narrow corridors that connect them. The narrow corridors correspond to edges of the dodecahedron and the triangular rooms correspond to vertices.
    }
    \label{fig:region}
\end{figure}

We will now show that if the lions are sufficiently far away in the initial situation, there exists a winning strategy for the man where he only takes stock of the situations in the quarters. That is, when he reaches a quarter, he must plan for the next $2$ units of time where to run to, and then he has reached a quarter again, and so on. 
\begin{invariant}\label{inv:sufficiently_far_away}
\begin{enumerate} In the scenario described above:
	\item \label{inv:sufficiently_far_away:1} The man is standing on a quarter.
	\item \label{inv:sufficiently_far_away:2} $\min\{\dnear,\dfar\}\geq 1$.
	\item \label{inv:sufficiently_far_away:3} At least one of the two following statements is true:
	\begin{itemize}
		\item $\dnear \geq 3$
		\item $\dfar \geq 7$
	\end{itemize}
\end{enumerate}
\end{invariant}

\begin{lemma}\label{lem:quarter_strategy}
	If Invariant~\ref{inv:sufficiently_far_away} is satisfied initially,
	the man has a winning strategy by which he runs from quarter to quarter at unit speed so that Invariant~\ref{inv:sufficiently_far_away} is true at any quarter. The strategy maintains Invariant~\ref{inv:sufficiently_far_away} Point~\ref{inv:sufficiently_far_away:2} at all times. In particular, the closest lion is always at distance at least $1$.
\end{lemma}

\begin{proof}
	Let $x$ denote the position of the man at the time $t$, and assume the invariant holds. We prove that he can run to another quarter $x'$ without getting caught such that the invariant again holds when he reaches $x'$.
		
	The proof goes by inspecting cases. Let $ab$ be the edge containing
	$x$ and suppose $a$ be the nearest vertex to $x$ and $b$ the furthest.
	
\emph{Case 1:} $\dfar(t) \geq 7$. Let $y$ denote the other quarter on the same edge as $x$. We claim that the man can run to $y$ without violating the invariant. We must thus argue that the invariants are satisfied at time $t+2$ for a man situated at $y$. 
First, note that he will not encounter any lion while running towards $y$ because $\dfar(t)>4$. Note also that $\dfar(t+2)\geq 1$, since $\dnear(t)\geq 1$ and the worst case is that the lion follows the man.
Furthermore, $\dnear (t+2) \geq \dfar(t)-4 = 3$, since $\dfar(t)\geq 7$ and the worst case is that the man and lion have run towards each other. Thus, the invariant holds at the time $t+2$.

\begin{figure}[h]
\vspace{-0pt}	    
\centering
\begin{minipage}[b]{.45\textwidth}
	\centering
     \includegraphics[trim={2.7cm 1cm 2cm 1cm},clip,width=\textwidth]{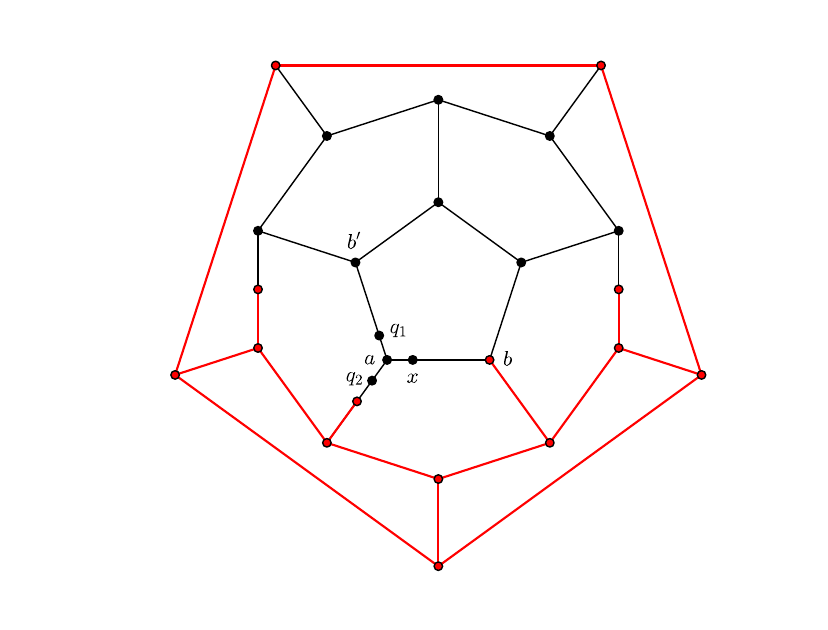}
     \caption{A situation from the proof of Lemma~\ref{lem:quarter_strategy}.
     Imagine that all edges have length $4$.
     The lion $\lnear$ is in the red part.}
     \label{fig:dodecahedron-symmetric}
\end{minipage}\hspace{0.10\textwidth}%
\begin{minipage}[b]{.45\textwidth}
	\centering
	\includegraphics[trim={-1.5cm 3cm 0cm 2.5cm},clip,width=\textwidth]{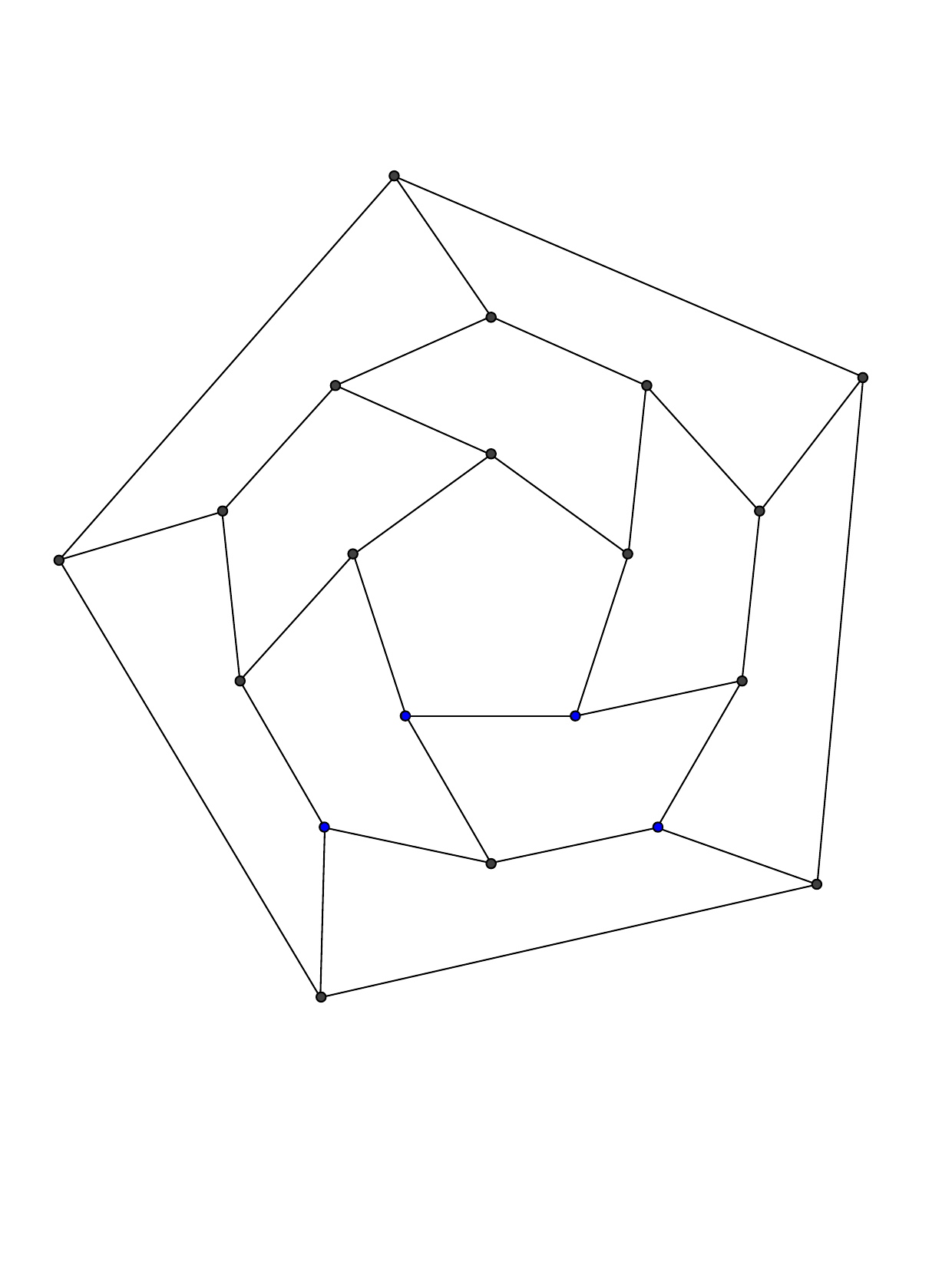}
	\caption{The embedding $\mathcal D$ of the dodecahedron.
	All edges have lengths $1$ or $3$.\\}
	\label{fig:dodecahedron}
\end{minipage}
\end{figure}

%
\emph{Case 2:} $\dfar(t) < 7$, and thus $\dnear(t) \geq 3$. In this case, we exploit the fact that $\dfar(t)$ is so small that we can bound $\dfar(t+2)$ from below. Let $\lfar$ denote the lion at distance $\dfar(t)$ from $x$ at time $t$, and let $\lnear$ denote the other lion. Consider the two other quarters at distance $1$ from $a$, call them $q_1$ and $q_2$. Assume without loss of generality that $q_1$ is furthest from $\lnear$. The situation is sketched in Figure~\ref{fig:dodecahedron-symmetric}.
We now argue that the man can choose to run towards $q_1$ without getting eaten, and while maintaining the invariant. 
Let $b'$ denote the vertex at distance $3$ to $q_1$.
Note that $d_{b'}(q_1,\lfar(t))\geq 11$
and thus, $d_{b'}(q_1,\lfar(t+2))\geq 9$.

In Figure~\ref{fig:dodecahedron-symmetric}, the points that are both
$\geq 3$ from $x$, and (weakly) closer to $q_2$ than to $q_1$, are marked with red,
and hence by our choice of $q_1$, $\lnear$ must be in the subset marked with red at
time $t$.
As is easily seen by inspection, $d_{b'}(q_1,\lnear(t))\geq 9$,
and thus $d_{b'}(q_1,\lnear(t+2))\geq 7$.
But then, $\dfar(t+2)\geq \min\{9, 7\} = 7$, and Invariant~\ref{inv:sufficiently_far_away}.\ref{inv:sufficiently_far_away:1}
and~\ref{inv:sufficiently_far_away:3} are maintained.

To see that Invariant~\ref{inv:sufficiently_far_away}.\ref{inv:sufficiently_far_away:2} is still maintained,
note that $d_a(q_1,\lnear(t))\geq 3$ and therefore
$d_a(q_1,\lnear(t+2))\geq 1$.
Similarly, since $d_b(x,\lfar(t))\geq 1$, we have
$d_a(q_1,\lfar(t))\geq 3$ so that
$d_a(q_1,\lfar(t+2))\geq 1$.
Thus, $\lnear(t+2)\geq 1$, and we are done. 
\end{proof}


Our first goal is to find an embedding $\mathcal G$
of the dodecahedron in the plane with
the properties described below, 
which will 
make it easier for us
to construct the region $R$.

\begin{lemma}\label{niceEmbedding}
	There exists a planar embedding $\mathcal G$
	of the dodecahedron such that
	\begin{itemize}
	\item
	all edges have length $4$,
	\item
	all edges consist of line segments with lengths being multiples of $\frac{1}{8}$,
	\item
	any pair of
	line segments from different edges that meet at a vertex each have length $\frac 14$
	and form an angle of size $\frac{2\pi}{3}$, and
	\item
	for any vertex $v$, the circle $D_v$ centered at $v$ with radius $\frac 1{16}$ only
	intersects the three edges incident to $v$.
	\end{itemize}
\end{lemma}

After proving this lemma, we derive from $\mathcal G$ a truly two-dimensional
area $R$ in the plane where the man can survive against two lions.
Lemma~\ref{lem:quarter_strategy} gives a winning
strategy for the man in $\mathcal G$ where he runs from quarter to quarter.
The paths along which he runs in $R$ will be
exactly the same as in $\mathcal G$ except for inside the circles $D_v$.

We first need the following elementary geometric observations:

\begin{observation}\label{obs:3-1-embedding}
	There exists a planar embedding $\mathcal D$ of the dodecahedron
	such that all edges have length $1$ or $3$.
	$\mathcal D$ furthermore has the property that the circle of radius $\frac{1}{4}$ centered at any vertex $v$ only intersects the three edges incident to $v$. (See Figure~\ref{fig:dodecahedron}.)
\end{observation}

\begin{lemma}\label{lem:triangle}
   For any three points $a,b,c$ on a circle $C$,
   there exist a equilateral triangle with corners $a',b',c'$ on $C$
   where $\{a,b,c\}$ and $\{a',b',c'\}$ are disjoint and such that, when considering
   the points $a,b,c,a',b',c'$ all together, $a$ is a neighbour of
   $a'$, and $b$ is a neighbour of $b'$, and $c$ is a neighbour of $c'$.
\end{lemma}

\begin{proof}
	See Figure~\ref{fig:vertexouter}.
	The points $a,b,c$ divide $C$ into three arcs.
	Clearly, we can choose an equilateral triangle with corners on $C$ disjoint
	from $\{a,b,c\}$ so that
	not all three corners of the triangle are on the same arc.
	It is now easy to label the corners
	of the triangle with $a',b',c'$ to satisfy the lemma.
\end{proof}

\begin{figure}[h]
\centering
\begin{minipage}[b]{.46\textwidth}
	\centering
	\includegraphics[trim={33cm 15cm 43cm 15cm},clip,width=0.8\textwidth]{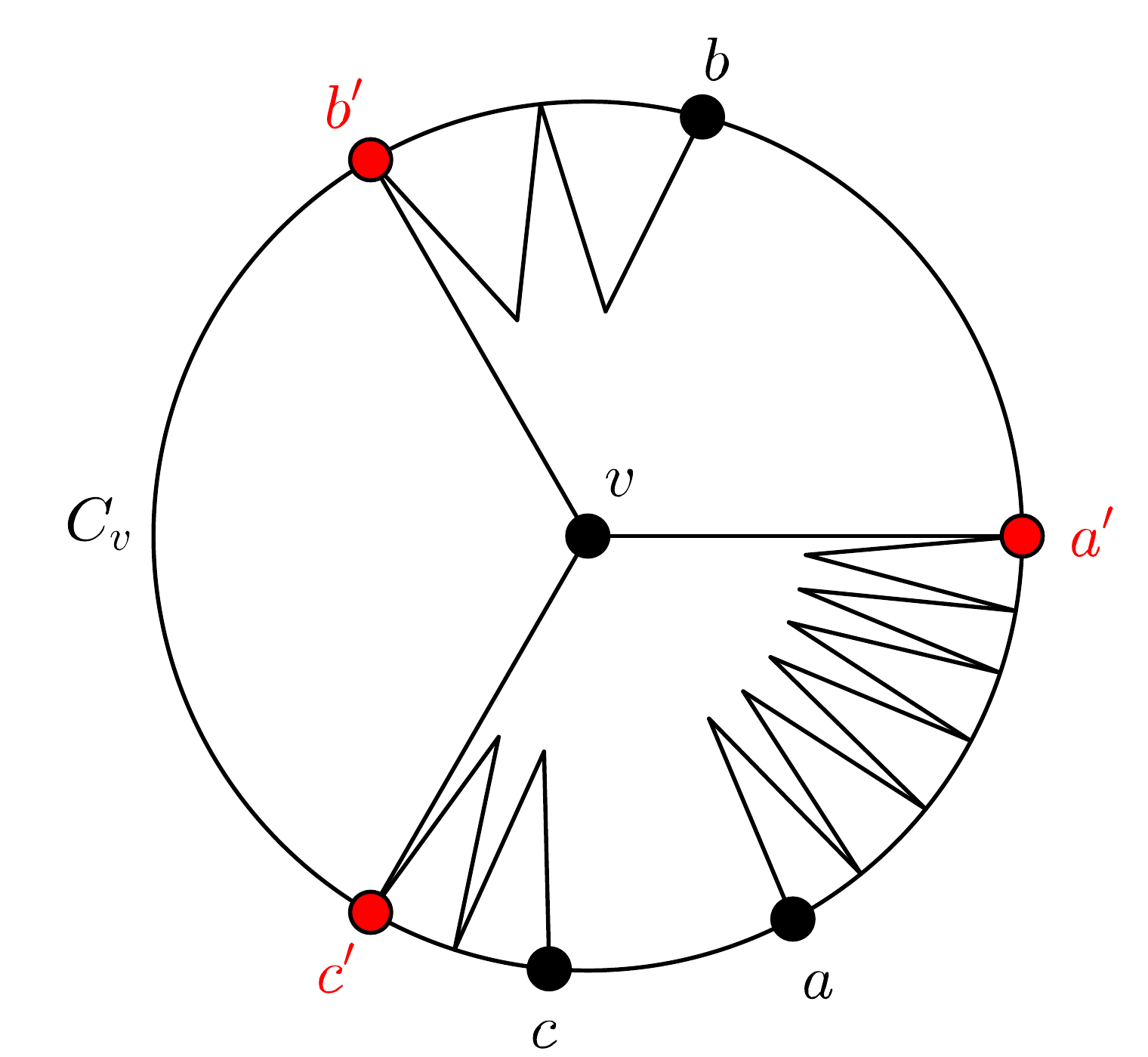}
	\caption{Regardless of angles between $a,b,c$,
	we can introduce bends to make the three edges meet at $v$
	in angles of size $\frac{3\pi}{2}$ and at the same time extend
	the lengths suitably.}
	\label{fig:vertexouter}
\end{minipage}\hspace{0.04\textwidth}%
\begin{minipage}[b]{.5\textwidth}
	\centering
\includegraphics[trim={31cm 10cm 32cm 10cm},clip,scale=0.22]{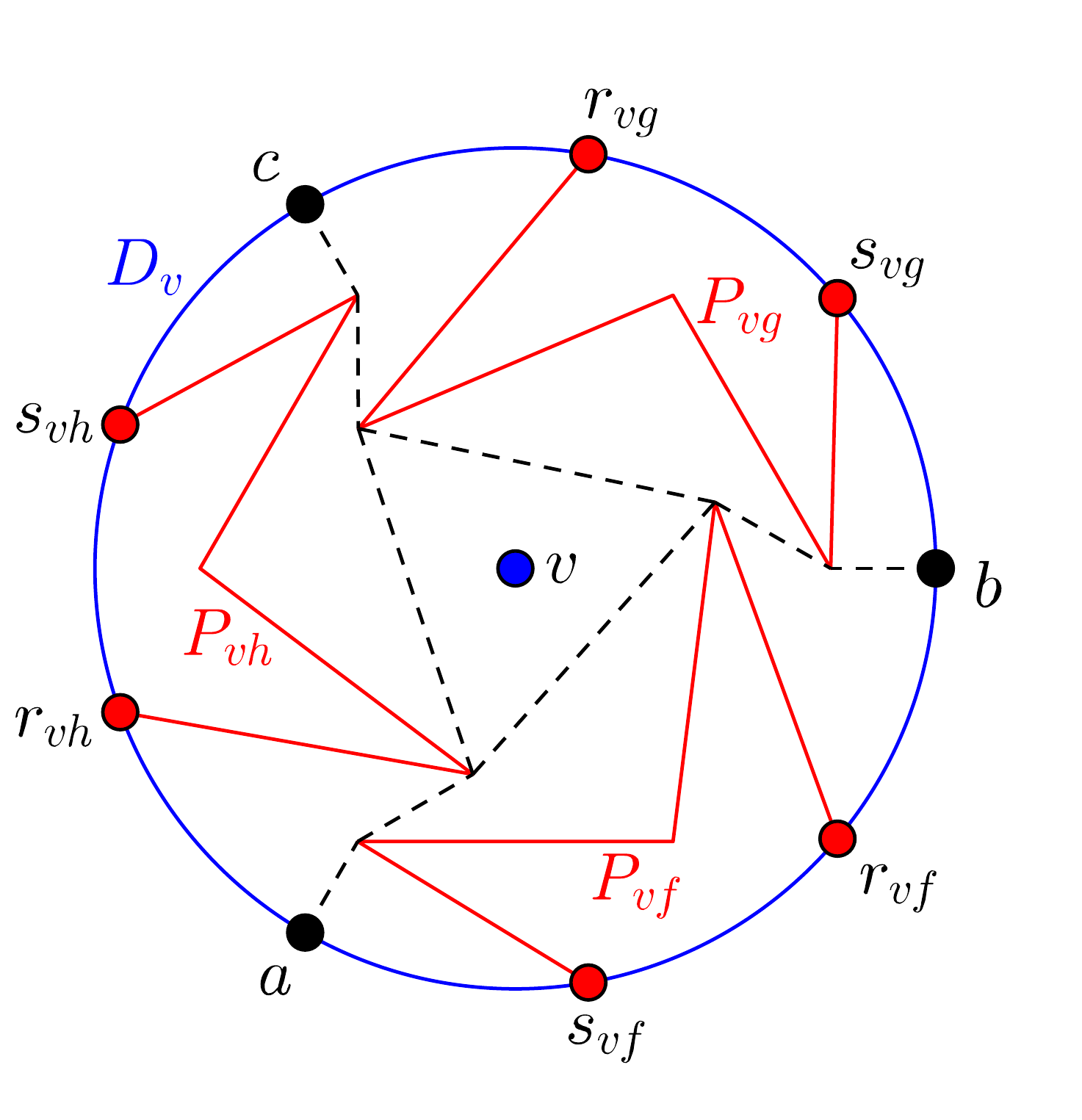}
\caption{The shortest paths in the circle $D_v$ between any two of $a,b,c,$ that
avoid crossing the polygonal curves $P_{vf},P_{vg},P_{vh}$ all have length
$1/8$.\\}
\label{shortestEqual}
\end{minipage}
\end{figure}

We are now ready to prove that a planar embedding $\mathcal G$
of the dodecahedron exists as stated in Lemma~\ref{niceEmbedding}.

\begin{proof}[Proof of Lemma~\ref{niceEmbedding}.]
Start with the embedding $\mathcal D$
shown in Figure~\ref{fig:dodecahedron}, where
all edges have length $1$ or $3$. 
Consider a vertex $v$ and the circle $C_v$ of radius $r=\frac{1}{4}$ centered
at $v$. Assume the three edges incident to $v$
enter $C_v$ in the points $a,b,c$, and let $u_a,u_b,u_c$ be the neighbouring vertices
of $v$
such that $a$ is a point on the edge $\{u_a,v\}$, $b$ is a point on $\{u_b,v\}$, and
$c$ is a point on $\{u_c,v\}$.
We now delete the segments $va$, $vb$, and $vc$, and therefore
need to reconnect $a$, $b$, and $c$ to $v$.
We explain how to reconnect $a$ to $v$; $b$ and $c$ are handled analogously.
We find points $a',b',c'$ on $C_v$
as described in Lemma~\ref{lem:triangle}. See Figure~\ref{fig:vertexouter}.
We first connect $a'$ to $v$.
We now need to connect $a$ to $a'$ using some bends.
A \emph{bend} is two segments $xy$
and $yz$, each
of length $r/2=1/8$, such that $x$ and $z$ are on $C_v$ and $y$ is in the interior of
$C_v$.
If the edge $\{u_a,v\}$ had length $3$ in $\mathcal D$,
we make two bends that together
connect $a$ and $a'$. We
thus increase the length of the edge $\{u_a,v\}$ by $1/2$ in each end and the resulting
edge has length $4$.
Note that two bends are just enough to connect $a$ and $a'$ in the worst case where the angle on $C_v$ between $a$ and $a'$ is almost $\frac {2\pi}3$, since each bend can span up to an angle of $\frac{2\pi}6$ on $C_v$.
If, on the other hand, the edge $\{u_a,v\}$ had length $1$ in $\mathcal D$, we connect
$a$ and $a'$ by $6$ bends, corresponding to extending the length of the edge
by $3$. The result is a planar embedding $\mathcal G$ of the dodecahedron with the
properties stated in the lemma.
\end{proof}

We now describe how to make the region $R$. We want each quarter of
$\mathcal G$ to be a point in $R$ and we want all pairs of quarters to have
the same distances in $\mathcal G$ and $R$. It will then follow from
Lemma~\ref{lem:quarter_strategy} that the man has a winning strategy by
running from quarter to quarter in $R$. We make one lake $L_f$ corresponding
to each face $f$ of $\mathcal G$.
Here, we also consider the outer boundary of
$R$ to be the boundary of an unbounded lake corresponding to the
exterior face of $\mathcal G$.
The
shortest paths in $R$ will be polygonal paths with corners at
convex corners of the lakes. Outside the circles
$D_v$, the paths along which the man will run are exactly the paths in
$\mathcal G$. Inside a circle $D_v$, we need to take special care to ensure
that the man can always run along an optimal path.

We now explain the construction of the lakes $L_f$
corresponding to
faces $f$ of $\mathcal G$.
Consider a vertex $v$ of $\mathcal G$ and the faces $f,g,h$ on which $v$ is a vertex.
We first describe how the boundaries of $L_f,L_g,L_h$ look in
the circle $D_v$ of radius $1/16$
centered at $v$. See Figure~\ref{shortestEqual}.
Let $a,b,c$ be the points where the edges incident to
$v$ enter $D_v$. Suppose that the arc on $D_v$ from
$a$ to $b$ is in the face $f$, the arc from $b$ to $c$ is in $g$, and the arc from
$c$ to $a$ is in $h$.
We now create three polygonal curves $P_{vf},P_{vg},P_{vh}$
inside $D_v$ so that the shortest path between any two of
$a,b,c$ contained in $D_v$ and not crossing any of $P_{vf},P_{vg},P_{vh}$
has length $1/8$.
The curve $P_{vf}$ starts at a point $r_{vf}$ on $D_v$ and
ends at a point $s_{vf}$ on $D_v$, and the endpoints
$r_{vf},s_{vf}$ are inside $f$, and similarly for the faces $g,h$.
These properties are easy to obtain by a construction
as shown in Figure~\ref{shortestEqual}.
The curves $P_{vf},P_{vg},P_{vh}$ will be part of the boundary of the lakes
$L_f,L_g,L_h$, respectively.

\begin{figure}
\centering
\includegraphics[trim={20cm 10cm 15cm 18cm},clip,width=\textwidth]
{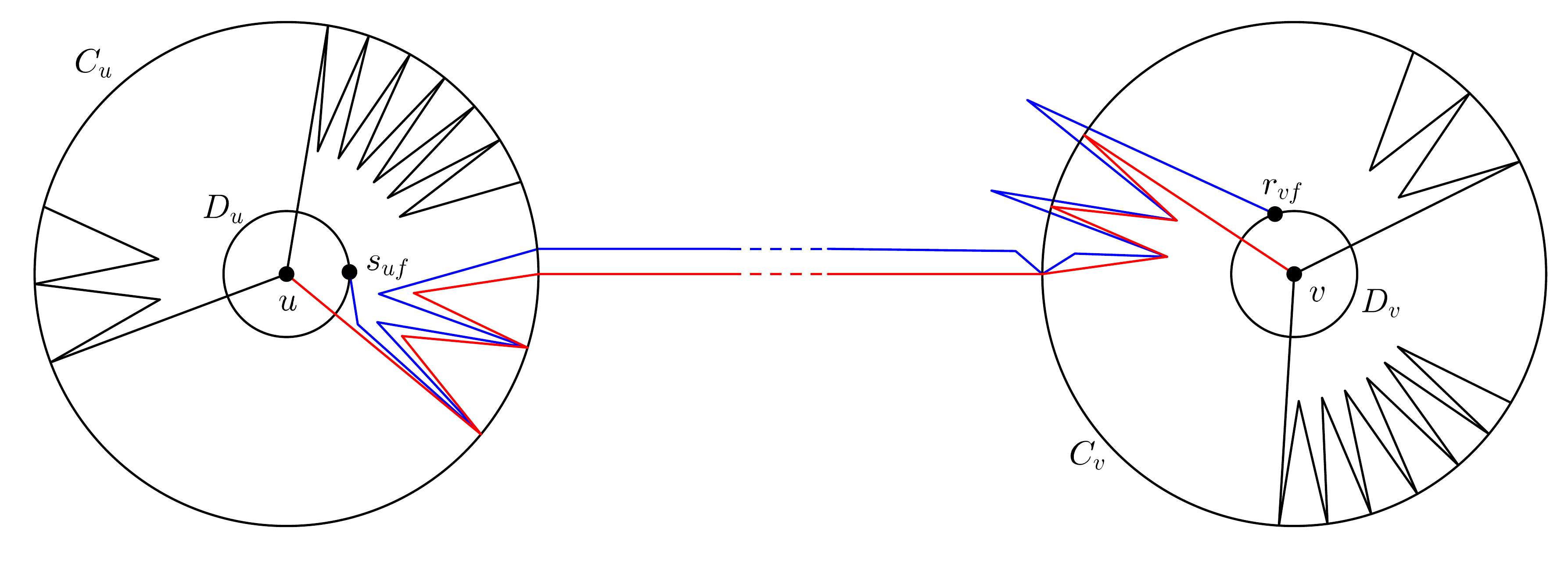}
\caption{The edge $e_{uv}$ of $\mathcal G$ is red and is
one of the edges bounding the face $f$, which is above
$e_{uv}$. The polygonal
curve $Q_{uv}$, which is on the boundary of the lake $L_f$, is blue.}
\label{connectFig}
\end{figure}

We now explain how to construct the rest of the boundary of each lake
$L_f$.
Consider a face $f$ of $\mathcal G$ and assume that the vertices on $f$ are $uvxyz$
in that order on $f$. 
The curves $P_{uf},P_{vf},P_{xf},P_{yf},P_{zf}$
appear on the boundary of $L_f$ in that order. In the following, we describe how to
connect the end $s_{uf}$ of $P_{uf}$ with the start $r_{vf}$ of $P_{vf}$ -- 
the other curves are connected in a completely analogous way.
See Figure~\ref{connectFig}.
Let $e_{uv}$ be the edge of $\mathcal G$ between $u$ and $v$, thus, 
$e_{uv}$ is a polygonal curve.
Let a \emph{corner} of $e_{uv}$ be a common point of
two neighbouring segments of $e_{uv}$.
We make a polygonal curve $Q_{uv}$ corresponding to $e_{uv}$.
$Q_{uv}$ starts at $s_{uf}$ and ends at $r_{vf}$
so that it connects $P_{uf}$ and $P_{vf}$.
$Q_{uv}$ stays near $e_{uv}$ inside $f$ and touches
$e_{uv}$ at the corners of $e_{uv}$
which are convex corners of $f$.
To summarize, $Q_{uv}$ has the following properties:
\begin{enumerate}
\item\label{qcond:1}
$Q_{uv}$ starts at $s_{uf}$ and ends at $r_{vf}$,

\item\label{qcond:2}
$Q_{uv}$ is completely contained in $f$,

\item\label{qcond:3}
$Q_{uv}$ is, except for the endpoints $s_{uf},r_{vf}$, outside the circles
$D_{u}$ and $D_v$,

\item\label{qcond:4}
$Q_{uv}$ and $Q_{u'v'}$ are completely disjoint for any ordered pair
$(u'v')\neq (u,v)$ so that $\{u',v'\}$ is an edge of $\mathcal G$, and


\item\label{qcond:6}
$Q_{uv}$ touches $e_{uv}$ at a point $p$ if and only if
$p$ is a corner of $e_{uv}$ which is a convex corner of $f$.

\end{enumerate}

Observe that $Q_{vu}$ (note: not $Q_{uv}$!)
touches $e_{uv}$ at the corners which are concave
corners of $f$, since those are convex corners of the neighbouring face
on the other side of $e_{uv}$.

\begin{theorem}
	There exists a polygonal region $R$ in the plane
	with holes where the exterior and interior boundaries
	are all pairwise disjoint and
	such that the man has a winning strategy against two lions.
\end{theorem}

\begin{proof}
$R$ is the region that we get by removing from $\RR^2$ the interior of each
of the lakes $L_f$. Thus, the boundary of each lake is included in $R$, so that $R$ is
a closed set. $R$ is also bounded because we remove the interior of the unbounded lake
corresponding to the exterior face of $\mathcal G$.
Note that any point on an edge $e_{uv}$ of $\mathcal G$ which is outside
the circles $D_u$ and $D_v$ is a point in $R$.
Since the quarters of $e_{uv}$ are outside the circles $D_u$ and $D_v$,
it follows that they are also points in $R$.
Furthermore, our construction ensures that
the distance in $R$ between any two
quarters is the same as in $\mathcal G$. 
Let $\mathcal G'$ be the points in $R$ which are on some shortest path between
two quarters in $R$.
Thus, $\mathcal G'$ are the points that the man can possibly visit when running
along shortest paths in $R$ from quarter to quarter.

Let $\lion_1$ and $\lion_2$ be two lions in $R$.
We define projections $\lion'_1$ and $\lion'_2$ of the lions $\lion_1$ and
$\lion_2$ to be the closest points in $\mathcal G'$
(with respect to distances in $R$).
We now define $\lion''_1$ and $\lion''_2$ to be projections of
$\lion'_1$ and $\lion'_2$ in $\mathcal G$ in the following way.
Outside the circles $D_v$, $\mathcal G$ and $\mathcal G'$ coincide, and here
we simply define $\lion''_i\mydef\lion'_i$. Suppose now that $\lion'_i$ is inside
a circle $D_v$ for some vertex $v$ of $\mathcal G$.
See Figure~\ref{fig:VertexProjection}.
Suppose that the three edges incident to $v$ enter $D_v$
at the points $a,b,c$. The projection
$\lion'_i$ is a point on one of the shortest paths between a pair
of the points $a,b,c$. Recall that these shortest paths
all have length $1/8$. Assume without loss of generality that
$\lion'_i$ is on the path from $a$ to $c$. Let $d$ be the distance
from $a$ to $\lion'_i$ in $R$, so that $0\leq d\leq 1/8$.
If $d=1/16$, we define $\lion''_i\mydef v$.
Otherwise, if $d<1/16$, we let $\lion''_i$ be the point on the segment
$av$ in $\mathcal G$ with distance $d$ to $a$, i.e.,
$\lion''_i\in av$ so that $\|a\lion''_i\|=d$. Similarly, if
$d>1/16$, we let $\lion''_i$ be the point on $bv$ with distance
$1/8-d$ to $b$.

We now prove that $\lion''_i$ moves with at most unit speed in $\mathcal G$.
It will then follow from Lemma~\ref{lem:quarter_strategy}
that the man has a winning strategy.

$\mathcal G'$ subdivides $R$ into some regions $R'_1,\ldots,R'_k$, which are
the connected components of $R\setminus\mathcal G'$.
Let $R_i=\overline{R'_i}$ be the closure of $R'_i$.
Now, $R=\bigcup_{i=1}^k R_i$.
Inside each circle $D_v$, there is a \emph{triangular} region
bounded by three segments from $\mathcal G'$. All other regions
are bounded by a polygonal curve $C\subset\partial L_f$ on
the boundary of some lake $L_f$ and a concave chain
$\mathcal H\subset\mathcal G'$.
Call such a region \emph{normal}.
If the lion $\lion_i$ is in a normal region $R_j$
with boundary $\partial R_j=C\cup\mathcal H$ as described before,
the projection $\lion'_i$ is on $\mathcal H$. It then
follows from the concavity of $\mathcal H$ that $\lion'_i$,
and thus also $\lion''_i$, moves continuously and
with at most unit speed.

However, when $\lion_i$ is inside a triangular region
in $D_v$, the projection $\lion'_i$ might jump from one segment of the triangle to another.
Suppose that the three edges incident to $v$ enter $D_v$
at the points $a,b,c$ as in Figure~\ref{fig:VertexProjection}.
Let $a'$ be the point where the shortest paths from $a$
to $b$ and $c$ separate and define $b'$ and $c'$ similarly.
Thus, the points $a'b'c'$ are the corners of the triangular region.
Suppose that $\lion'_i$ jumps from $a'b'$ to $a'c'$. Then,
the distance from $\lion_i$ to $a'b'$ and $a'c'$ is the same
and the distance from $a$ to $\lion'_i$ before and after the
jump is at most $1/16$, since otherwise, $\lion_i$ would be
closer to the segment $b'c'$ than to $a'b'$ and $a'c'$.
It follows that $\lion'_i$ jumps
from one point to another which have the same projection
$\lion''_i$. Thus, $\lion''_i$ moves continuously and with
at most unit speed.

The man now employs the strategy from Lemma~\ref{lem:quarter_strategy}
in the following way. He imagines that he is playing in
the dodecahedron $\mathcal G$ against the lions $\lion''_1$
and $\lion''_2$. Assume therefore that Invariant~\ref{inv:sufficiently_far_away} holds initially.
The strategy tells the man to which neighbouring quarter to
run. That quarter also exists in $\mathcal G'$, and has
the same distance, so the man runs to that quarter in $\mathcal G'$.
Since $\lion''_1$ and $\lion''_2$ run with at most unit speed,
the man can escape them forever.
When the man is outside the circles $D_v$,
it is a necessary condition for the lions to catch the
man that $\lion''_1$ or $\lion''_2$ coincide with the man,
so we conclude that they cannot catch him outside the circles.
When the man is inside a circle $D_v$, we know
from Lemma~\ref{lem:quarter_strategy} that $\lion''_1$ and
$\lion''_2$ are at least $1$ away from the man. Therefore,
$\lion_1$ and $\lion_2$ must be outside $D_v$, and hence
they cannot catch him in that case either.
Thus, the man survives forever in $R$.
\end{proof}



\begin{figure}[h]
	\vspace{-10pt}	    
	\centering
		\includegraphics[trim={0.5cm 1cm 1cm 1cm},clip,width=0.45\textwidth]{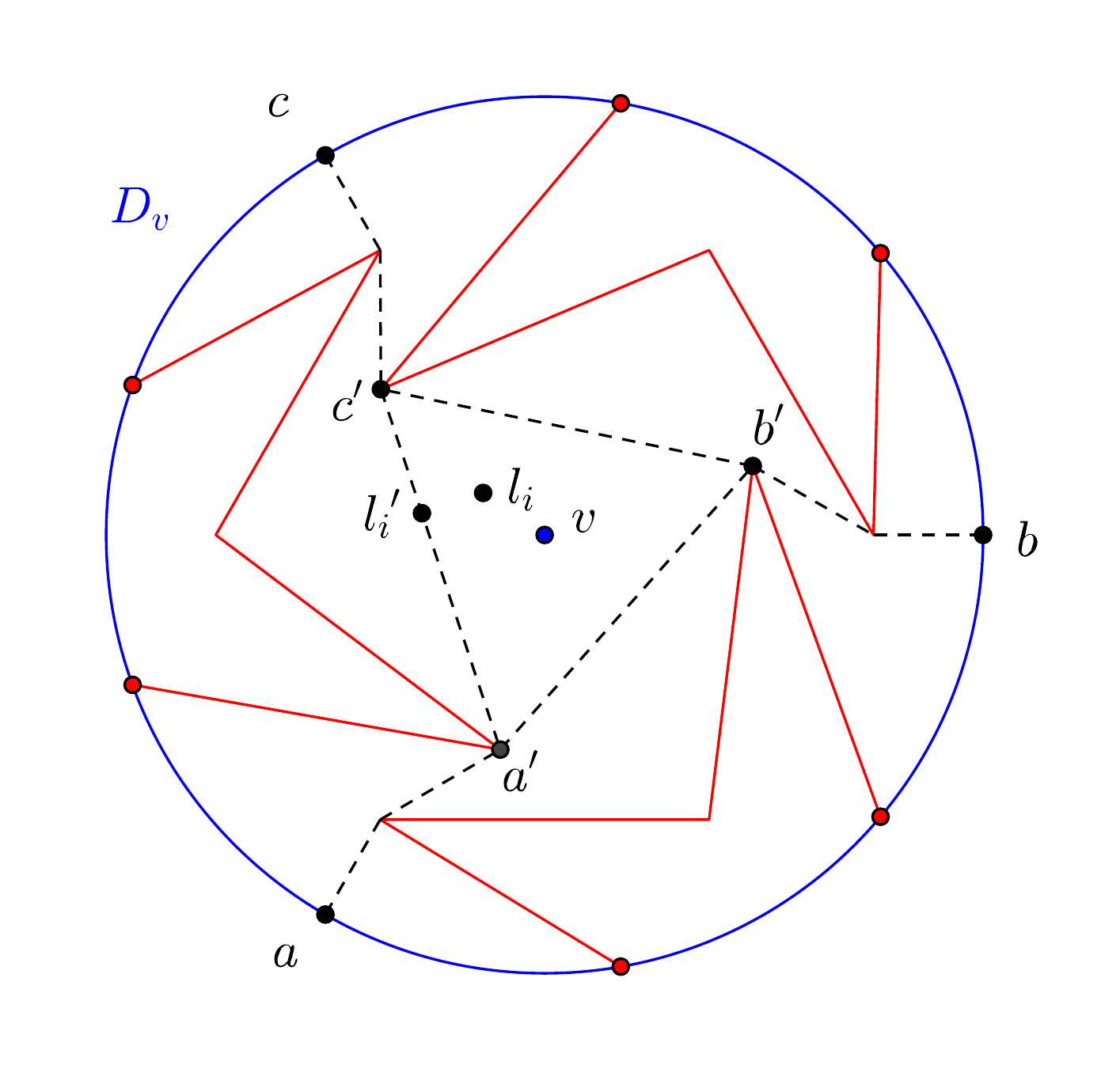}\hspace{0.09\textwidth}\includegraphics[trim={0.5cm 1cm 1cm 1cm},clip,width=0.45\textwidth]{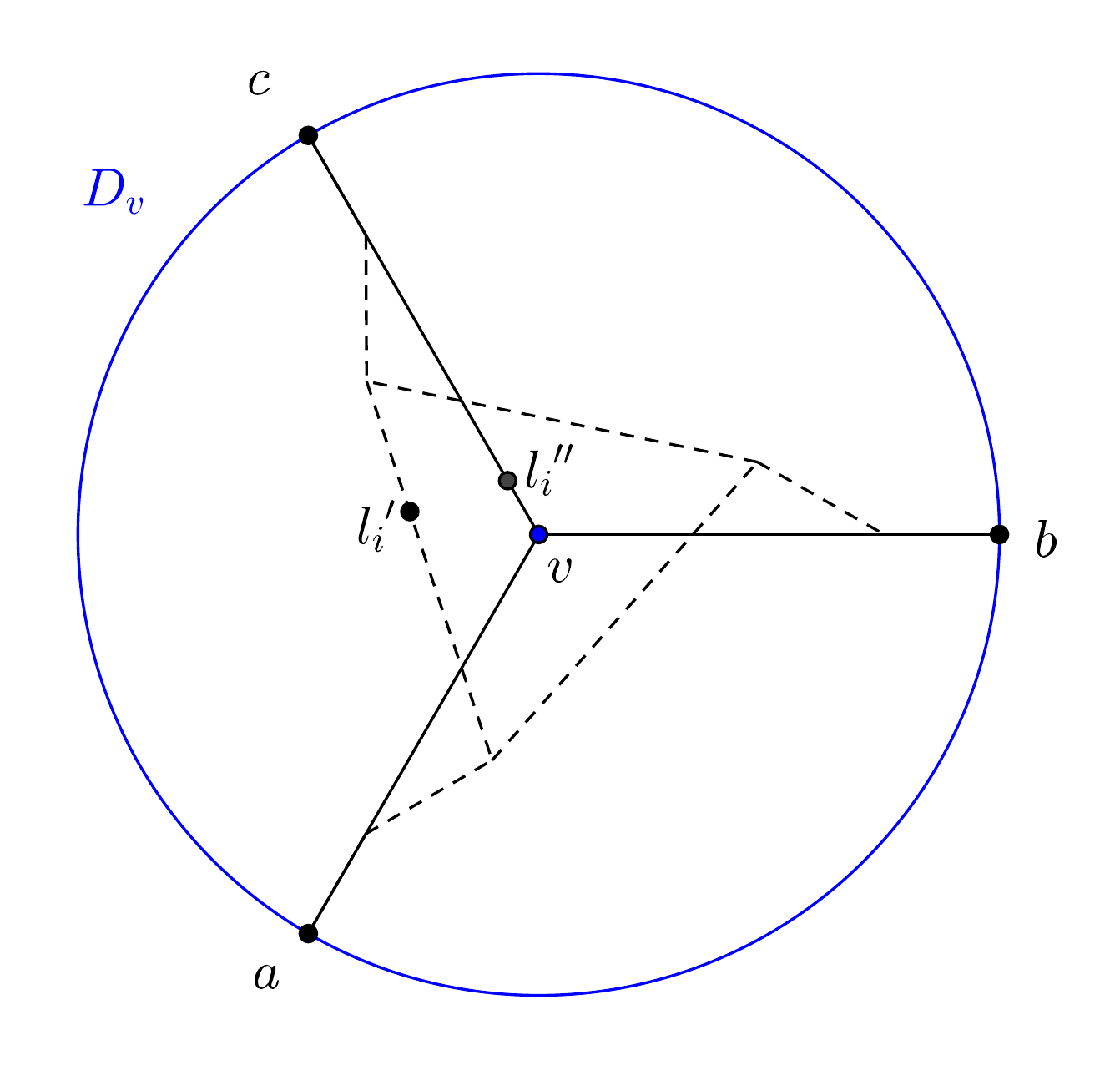}
		\caption{The projection of the lion's position $\lion_i$
    onto the point $\lion'_i$ of $\mathcal{G'}$ (left), and the projection of $\lion'_i$ onto the point $\lion_i''$ of $\mathcal{G}$ (right). The dashed lines illustrate $\mathcal{G'}$, and the solid lines illustrate $\mathcal{G}$.
    In the left figure, $\lion'_i$ is the closest point on $\mathcal G'$ to
    $\lion_i$.
    In the right figure,
    the length of the segment $c\lion''_i$ equals
    the length of the dashed path from $c$ to $\lion'_i$.}
		\label{fig:VertexProjection}
\end{figure}

\section{The Fast Man Surviving any Finite Number of Lions}\label{fastManSec}

We now consider the case where the man is just slightly faster than the lions in the unbounded plane without obstacles. In this case, the man is able to escape arbitrarily many lions.

\begin{theorem}\label{thm:fastMan}
    In the plane $\mathbb{R}^2$, for any $\ee>0$, a man able to run at speed $1+\ee$ has a locally finite strategy
    to escape the convex
    hull of any number $n\in\NN$ of unit-speed lions, provided that the man
    does not start at the same point as a lion. Thus, the man
    has a locally finite winning strategy.
\end{theorem}

In fact, we prove that the man is able to keep some minimum distance $d_{\ee,n}$ to any lion, where $d_{\ee,n}$ only depends on $\ee$, $n$, and the initial distances to the lions. Thus, if the $n$ lions and man are disks with radius $<\frac{1}{2} d_{\ee,n}$, the man is still able to escape.


We proceed by induction on the number $n$ of lions.
We define strategies $M_j$, $j=1,\ldots,n$, for the man so that following strategy $M_j$ keeps the man at a safety distance $c_{jk}>0$ from the $k$'th lion, $k\leq j$.
Each strategy yields a curve consisting of line segments all of the same length.

Inductively, the man can keep a safety distance $c_{n-1,j}$ to each lion $\lion_j$ for $j=1,\ldots,n-1$ by running at speed $1+\ee_{n-1}$, where $0<\ee_1<\ee_2<\ldots<\ee_n<\ee$.
The corners of the curve defined by strategy $M_{n-1}$ are milestones that he runs towards when using the strategy $M_n$, i.e., when also taking care of avoiding the $n$'th lion $\lion_n$.
When using strategy $M_n$, the man runs with speed $1+\ee_n$, i.e., slightly faster than when using strategy $M_{n-1}$.
This gives time to make some detours caused by the $n$'th lion $\lion_n$ while still being close to the milestone prescribed by strategy $M_{n-1}$.
If $\lion_n$ gets too close, the man makes an \emph{avoidance move}, keeping a safety distance $c_{nn}$ to $\lion_n$ which is much smaller than $c_{n-1,j}$ for any $j=1,\ldots,n-1$.
Intuitively, when performing avoidance moves, the man runs counter-clockwise around a fixed-radius circle centered at the lion. 
After a limited number of avoidance moves, the man can make an \emph{escape move}, where he simply runs towards the milestone defined by the strategy $M_{n-1}$.

By choosing $c_{nn}$ sufficiently small, we can make sure that the detour caused by the $n$'th lion is so small that it can only annoy the man once for each of the segments of the strategy $M_{n-1}$, and thus that he is ensured to be very close to the position defined by strategy $M_{n-1}$ and hence not in danger of the first $n-1$ lions.
Thus, the safety distance $c_{nj}$ to a lion $\lion_j$ for $j=1,\ldots,n-1$, is only slightly smaller than $c_{n-1,j}$.

\begin{proof}[Proof of Theorem~\ref{thm:fastMan}]
We assume without loss of generality that $\ee<1$.
Let $\lion_1,\ldots,\lion_n$ be $n$ arbitrary lion paths and let the man start
at position $\man_0$ such that $\man_0\neq \lion_i(0)$ for all $i$.
We show that the
man has a strategy $\manStrat_n$ with the following properties:
\begin{enumerate}
\item \label{prop:2}
The man is always running at speed $1+\ee_n$, where
$\ee_n\mydef(1-2^{-n})\cdot\ee$.

\item \label{prop:3}
The path defined by
$\manStrat_n(\lion_1,\ldots,\lion_n)$ is a polygonal path with corners
$m_0m_1\ldots$ and
each segment $m_im_{i+1}$ has the same length $\deltaT_n\cdot(1+\ee_n)$.
Thus, the time it takes the man to run from $m_i$ to $m_{i+1}$ is
$\deltaT_n$.

\item \label{prop:3a}
Let $t_i\mydef i\cdot\deltaT_n$ be the time where the man leaves
$m_i$ in order to run to $m_{i+1}$. The point $m_{i+1}$ can be determined
from the positions of the lions at time $t_i$.

\item \label{prop:4}
For any lion $\lion_j$, there exists a \emph{safety distance}
$c_{nj}>0$ such that for any $i=1,\ldots$,
any $t\in [t_{i},t_{i+1}]$, and any point $x\in m_{i}m_{i+1}$, it holds that
$\|x\lion_j(t)\|\geq c_{nj}$.

\item \label{prop:1}
There is a corner $m_i=\manStrat_n(t_i)$ such that for all $t\geq t_i$,
$$\manStrat_n(\lion_1,\ldots,\lion_n)(t)\notin\CH\{\lion_1(t),\ldots,\lion_n(t)\}.$$
\end{enumerate}

Clearly, it follows from the properties that $M_n$ is a
winning strategy for the man
fulfilling the requirements in the theorem.
We prove the statement by induction on $n$.
If there is only one lion, the man will run on the same ray all the time with constant
speed $1+\ee_1=1+\ee/2$. The man chooses the
direction of the ray to be $m_0-\lion_1(0)$.
This strategy obviously satisfies the stated properties.
Assume now that a strategy $M_{n-1}$ with the stated properties
exists for $n-1\ge 1$ lions and
consider a situation with $n$ lions running along paths
$\lion_1,\ldots,\lion_n$.

For any $i\in\{1,\ldots,n\}$, let
$\manStrat_i$ be shorthand for $\manStrat_i(\lion_1,\ldots,\lion_i)$ and $\man$ shorthand for $\manStrat_n$.

At any time $t$, let the succeeding corner on the strategy $\manStrat_{n-1}$ be
$$\goal(t)\mydef \manStrat_{n-1}(\lfloor t/\deltaT_{n-1}+1\rfloor\cdot\deltaT_{n-1}).$$
By property~\ref{prop:3a}, the man can always compute the point
$\goal(t)$.

We first describe the intuition behind the man's strategy
without specifying all details,
and~later give a precise description. 
In the situation with $n$ lions,
the man attempts to run according to the strategy for the $n-1$ first lions, i.e.,
the strategy $\manStrat_{n-1}$.
Thus, at any time $t$, the man's goal is to run towards the point $\goal(t)$.
However, the lion $\lion_n$ might prevent him from doing so. Compared to the
case with $n-1$ lions, the man has increased his speed by
$1+\ee_n-(1+\ee_{n-1})=2^{-n}\ee$, so
he has time to take detours while still following the strategy $\manStrat_{n-1}$
approximately.

Assume that we have defined the man's strategy up to time $t$.
If he is close to the
$n$'th lion, i.e., the distance $\|\man(t)\lion_n(t)\|$ is close to $\rD$,
for some small constant $\rD>0$ to be specified later, he runs counterclockwise
around the lion, maintaining approximately distance $\rD$ to the lion.
He does so until he gets to a point where running directly towards $\goal(t)$ will not
decrease his distance to the lion.
He then escapes from the lion, running directly towards
$\goal(t)$. Doing so, he can be sure that
the lion cannot disturb him anymore until he reaches $\goal(t)$ or
$\goal(t)$ has changed.

We choose $\rD$ so small that when the man is running around the lion,
we are in one of the following cases:
\begin{itemize}
\item
The lion is so close to $\goal(t)$ that the man is within the safety distance $c_{n-1,j}$ from
$\goal(t)$ for each $j=1,\ldots,n-1$, and thus in no danger of the lions $\lion_1,\ldots,\lion_{n-1}$.

\item
After running around the lion in a period of time no longer than $12\pi \rD/\ee_n$,
the man escapes by running directly towards $\goal(t)$ without decreasing
the distance to the lion. By choosing $\rD$ sufficiently small,
we can therefore limit the duration, and hence the length, of the detour that
the lion can force the man to run, so that the man is ensured to
be within the safety distance from the lions $\lion_1,\ldots,\lion_{n-1}$ during the detour.
\end{itemize}

We now describe the details that make this idea work.
Let
$$C_n\mydef \min_{j=1,\ldots,n-1} \frac{c_{jj}}{2^{n-j+1}}.$$
Informally, $C_n$ is the distance that we allow the lions $\lion_1,\ldots,l_{n-1}$ to get closer to the man in the strategy $M_n$ as compared to $M_{n-1}$.
We define
$$\rD\mydef\min\left\{
\frac{\deltaT_{n-1}\ee_n(\ee_n-\ee_{n-1})}{2+2\ee_n+18\pi(1+\ee_n)},
\frac{\ee_nC_n}{2+2\ee_n+12\pi(1+\ee_n)},
\|m_0\lion_n(0)\|
\right\},$$
$$\rhoD\mydef 2\rD/\ee_n,$$
$$\theta\mydef\arccos\frac 1{1+\ee_n},$$
$$\varphi\in(0,\pi/2]\quad\text{so that}\quad\tan\theta=\frac{\rhoD \sin\varphi}{\rhoD\cos\varphi-2\rD},\quad\text{and}$$
$$\deltaT_n>0\quad\text{so that}\quad 2\arcsin\frac{(1+\ee_n)\deltaT_n}{2(\rD-\deltaT_n)}+\frac{\deltaT_n}{\rhoD}\leq\varphi,
\quad \deltaT_n<\frac \rD{3+\ee_n},
\quad\text{and}\quad \deltaT_{n-1}/\deltaT_n\in\NN.$$

We note that $\varphi$ can be chosen since the function
$x\longmapsto \frac{\rhoD \sin x}{\rhoD\cos x-2\rD}$ is $0$ for $x=0$ and
tends to $+\infty$ as $\rhoD\cos x$ decreases to $2r$.
As for $\deltaT_n$, the function
$x\longmapsto 2\arcsin\frac{(1+\ee_n)x}{2(\rD-x)}+\frac{x}{\rhoD}$
is $0$ for $x=0$ and increases continuously, and hence $\deltaT_n$ can be chosen.

\begin{figure}
\centering
\begin{minipage}[t]{.3\textwidth}
  \centering
  \includegraphics[scale=0.8]{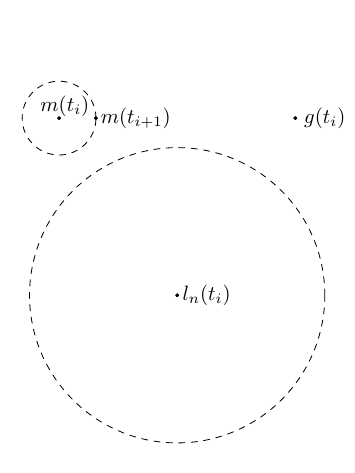}
  \captionof{figure}{A free move. The circles with centers
  $\man(t_i)$ and $\lion_n(t_i)$ have radii $(1+\ee_n)\deltaT_n$ and
  $\rD$, respectively.}
  \label{freeMove}
\end{minipage}\hspace{0.04\textwidth}%
\begin{minipage}[t]{.3\textwidth}
  \centering
  \includegraphics[scale=0.8]{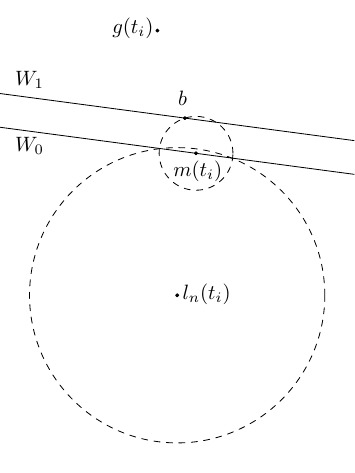}
\captionof{figure}{An escape move. The man runs to $b$.}
\label{escapeMove}
\end{minipage}\hspace{0.04\textwidth}
\begin{minipage}[t]{.3\textwidth}
  \centering
  \includegraphics[scale=0.8]{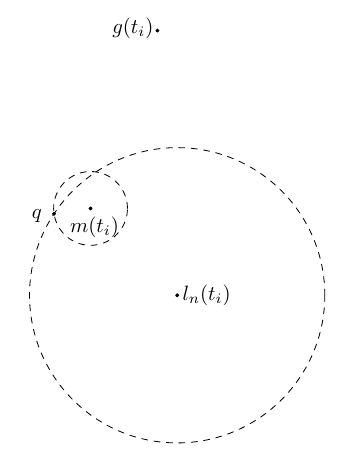}
  \captionof{figure}{An avoidance move. The man runs to $q$.}
  \label{avoidanceMove}
\end{minipage}%
\end{figure}

Define a point in time $t$ to be a \emph{time of choice} if
$t$ has the form $t_i\mydef i\deltaT_n$ for $i\in\NN_0$.
At any time of choice $t_i$,
the man chooses the point $\man(t_{i+1})$
at distance $(1+\ee_n)\deltaT_n$ from
his current position $\man(t_i)$ by the following strategy
(see Figures~\ref{freeMove}--\ref{avoidanceMove}):
\begin{enumerate}[label=\textbf{\Alph*}]
\item\label{choice:a}
Suppose first that $\|\man(t_i)\lion_n(t_i)\|\geq \rD+\deltaT_n(1+\ee_n)$.
Then the man chooses the
direction directly towards $\goal(t_i)$.
In the exceptional case that $\man(t_i)=\goal(t_i)$, he chooses an arbitrary direction.

\item\label{choice:b}
Suppose now that $\|\man(t_i)\lion_n(t_i)\|< \rD+\deltaT_n(1+\ee_n)$ and consider the case
$\man(t_i)\neq \goal(t_i)$.
Let $b$ be the point at distance $(1+\ee_n)\deltaT_n$ from $\man(t_i)$ in the direction
towards $\goal(t_i)$.
If there exist two parallel lines $W_0$ and $W_1$ such that
$\man(t_i)\in W_0$, $b\in W_1$, $\dist(\lion_n(t_i),W_0)\geq \rD-\deltaT_n$,
and $\dist(\lion_n(t_i),W_1)\geq \dist(\lion_n(t_i),W_0)+\deltaT_n$,
then the man runs to $b$.

\item\label{choice:c}
In the remaining cases,
the circles $C(\man(t_i),\deltaT_n(1+\ee_n))$
and $C(\lion_n(t_i),\rD)$ intersect at two points $p$ and $q$ such that
the arc on $C(\lion_n(t_i),\rD)$ from $p$ counterclockwise to $q$ is in the interior of
$C(\man(t_i),\deltaT_n(1+\ee_n))$. The man then runs towards the point
$q$.
\end{enumerate}

A move defined by case~\ref{choice:a},
\ref{choice:b}, or \ref{choice:c} is
called a
\emph{free move}, an \emph{escape move}, or an \emph{avoidance move},
respectively.
Let \emph{move $i$} be the move that the man does during the interval
$[t_i,t_{i+1})$.

\begin{claim}\label{claim1}
At any time of choice $t_i$, it holds that
$\|\man(t_i)\lion_n(t_i)\|\geq \rD-\deltaT_n$
and if the preceding move was an avoidance move, it also holds that
$\|\man(t_i)\lion_n(t_i)\|\leq \rD+\deltaT_n.$
Furthermore, at an arbitrary point in time $t\in [t_{i-1},t_i]$ and any point
$\man'\in \man([t_{i-1},t_i])$ it holds that
$0<\rD-(3+\ee_n)\deltaT_n\leq\|\man'\lion_n(t)\|$
and
if move $i-1$ is an avoidance move then additionally
$\|\man'\lion_n(t)\|\leq \rD+(3+\ee_n)\deltaT_n.$
\end{claim}

\begin{proof}
We prove the claim by induction on $i$. It clearly holds for $i = 0$ so assume that it holds for $i-1$. If move $i-1$ was a free move,
we have
\begin{align*}
\|\man(t_i)\lion_n(t_i)\|&\geq
\|\man(t_{i-1})\lion_n(t_{i-1})\|-(2+\ee_n)\deltaT_n\\
&\geq
\rD+\deltaT_n(1+\ee_n)-(2+\ee_n)\deltaT_n=\rD-\deltaT_n.
\end{align*}

If move $i-1$ was an escape move, we have
$$\|\man(t_i)\lion_n(t_{i-1})\|\geq
\|\man(t_{i-1})\lion_n(t_{i-1})\|+\deltaT_n.$$
Hence
\begin{align*}
\|\man(t_i)\lion_n(t_i)\|&\geq \|\man(t_i)\lion_n(t_{i-1})\|-\deltaT_n\\
&\geq \|\man(t_{i-1})\lion_n(t_{i-1})\|\geq \rD-\deltaT_n
\end{align*}

If move $i-1$ was an avoidance move, we have
$$\|\man(t_i)\lion_n(t_i)\|\geq \|\man(t_i)\lion_n(t_{i-1})\|-\deltaT_n
= \rD-\deltaT_n
$$
and, similarly,
$$\|\man(t_i)\lion_n(t_i)\|\leq \|\man(t_i)\lion_n(t_{i-1})\|+\deltaT_n
= \rD+\deltaT_n.
$$

Since at a point of choice $t_{i-1}$
we have $\rD-\deltaT_n\leq\|\man(t_{i-1})\lion_n(t_{i-1})\|$ and
the lion and the man can move at most $(2+\ee_n)\deltaT_n$ closer
to each other within $\deltaT_n$ time,
it holds for any point in time $t\in[t_{i-1},t_i]$ and
any point $\man'\in \man([t_{i-1},t_i])$ that
\begin{align*}
\rD-(3+\ee_n)\deltaT_n\leq
\|\man(t_{i-1})\lion_n(t_{i-1})\|-(2+\ee_n)\deltaT_n\leq
\|\man'\lion_n(t)\|.
\end{align*}

If move $i-1$ is an avoidance move, we have
$\|\man(t_{i-1})\lion_n(t_{i-1})\|\leq \rD+\deltaT_n$,
so it holds for any point in time $t\in[t_{i-1},t_i]$ and
any point $\man'\in \man([t_{i-1},t_i])$ that
\begin{align*}
\|\man'\lion_n(t)\|
\leq
\|\man(t_{i-1})\lion_n(t_{i-1})\|+(2+\ee_n)\deltaT_n
\leq
\rD+(3+\ee_n)\deltaT_n.
\end{align*}

\end{proof}

\begin{claim}\label{claim2}
An avoidance move is succeeded by an avoidance move or an escape
move. When the man does an escape move, he will not do an avoidance move before
he reaches $\goal(t)$ or $\goal(t)$ moves.
\end{claim}

\begin{proof}
Consider move $i$.
We know from Claim~\ref{claim1} that if move $i-1$ was an avoidance move, then
$\|\man(t_i)\lion_n(t_i)\|\leq \rD+\deltaT_n<\rD+(1+\ee_n)\deltaT_n$,
so move $i$ cannot be a free move.

\begin{figure}
\centering
\includegraphics[scale=0.8]{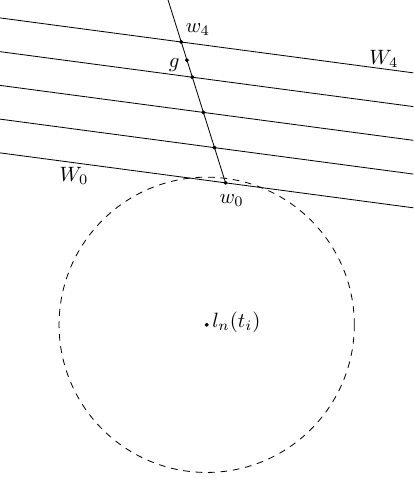}
\caption{The distance between two consecutive of
the parallel lines $W_0,\ldots,W_4$ is at least $\deltaT_n$, which
proves that the man runs from
$\man(t_i)=w_0$ to $w_4$ unless $g$ moves in the meantime.}
\label{claim3proof}
\end{figure}

For the second part of the statement, assume that move $i$ is an escape move.
Let $\goal\mydef \goal(t_i)$.
Let $w_0,\ldots,w_k$ be a sequence of points on the ray from $\man(t_i)$
with direction to $\goal$ such that
$w_0=\man(t_i)$, $\|w_0w_j\|=j(1+\ee_n)\deltaT_n$,
and $k$ is minimum such that either $\goal\in w_{k-1}w_k$ or
$\goal(t')\neq g$ for some $t'\in [t_{i+k-1},t_{i+k}]$.
See Figure~\ref{claim3proof}.
Let $W_0$ and $W_1$ be the parallel lines defined in case~\ref{choice:b} for
move $i$. We define lines $W_j$ for $j\geq 2$ to be parallel
to $W_0$ and passing through $w_j$. We claim that for any
$j\in\{0,\ldots,k-1\}$, the man moves from $w_j$ to $w_{j+1}$ during
move $i+j$ using either an escape move or a free move. We prove this by
induction on $j$. It holds for $j=0$ by assumption, so assume it holds that
$\man(t_{i+j})=w_j$ and that move $i+j-1$ was an escape move or a free move.
Since the distance between consecutive lines $W_j$ and $W_{j+1}$ is
at least $\deltaT_n$, it follows that
$\dist(\lion_n(t_i),W_j)\geq \rD+(j-1)\deltaT_n$ and hence that
$\dist(\lion_n(t_{i+j}),W_j)\geq \rD-\deltaT_n$.
Now, if $\|\man(t_{i+j})\lion_n(t_{i+j})\|< \rD+\deltaT_n(1+\ee_n)$,
then the lines $W_j$ and $W_{j+1}$ are a witness that move $i+j$ is an
escape move so that the man moves to $w_{j+1}$.
Otherwise, move $i+j$ is a free move, in which case the man moves
to $w_{j+1}$. Finally, since $\goal(t)$ moves or the man reaches $\goal$
during move $i+k$, the statement holds.
\end{proof}

Define $\rhoD'\mydef \rhoD+\rD+(3+\ee_n)\deltaT_n$
and $\tau\mydef 6\pi \rD/\ee_n$.

\begin{claim}\label{claim3}
If move $i$ is an avoidance move,
one of the following three events occurs before $\tau$ time
has passed: (i) $\goal(t)$ moves, (ii) $\|\man(t)\goal(t)\|<\rhoD'$, or (iii) the man makes an escape move.
\end{claim}

\begin{proof}
We first present an informal description of the proof.
If the first two events do not occur, it follows from Claim~\ref{claim2} that the man keeps doing avoidance moves during this time. Let $\xi(t)$ resp.~$\eta(t)$ denote the angle of the vector $\overrightarrow{\lion_n(t)\man(t)}$ resp.~$\overrightarrow{\lion_n(t)g(t_i)}$. A key observation is that if the difference in these angles is small, the man makes an escape move since then the lion and the goal $\goal$ are roughly on opposite sides of the man. Showing that this difference eventually becomes small involves showing that $\eta$ increases by at least $2\pi$ more than $\xi$ after $\tau$ time so that at some point in time $t\in[t_i,t_i+\tau]$, vectors $\overrightarrow{\lion_n(t)\man(t)}$ and $\overrightarrow{\lion_n(t)g(t_i)}$ have the same orientation. By Claim ~\ref{claim1}, the lion $\lion_n$ never gets closer than $\rho$ to $g(t_i)$ which implies that the change in $\eta$ is small in any time interval $[t_j,t_{j+1}]$. Since the man keeps a minimum distance to the lion, it similarly follows that the change in $\xi$ is small in $[t_j,t_{j+1}]$. Picking $j$ to be the maximum such that $t_j\le t$ gives $t - t_j\le\deltaT_n$ which implies that the difference in the two angles is small at time $t_j$ at which point the man makes an escape move. Since $t_j \le t + \tau$, the lemma follows.

We now proceed with the formal proof of the claim.
Assume that neither the first nor the second event occurs
before $\tau$ time has passed.
We therefore know by Claim~\ref{claim2}
that the man keeps doing avoidance moves during that time.
Let $\goal\mydef \goal(t_i)$.
By Claim~\ref{claim1}, we know that
for any $t\in[t_i,t_i+\tau]$, it holds that
$$\|\lion_n(t)\goal\|\geq \|\man(t)\goal\|-\|\man(t)\lion_n(t)\|\geq
\rhoD'-(r+(3+\ee_n)\deltaT_n)=
\rhoD.$$
There exist continuous functions
$\xi,\eta\colon\RR_0\longrightarrow\RR$
that measure the angle from $\lion_n$ to $\man$ and to $\goal$,
respectively, i.e., such that for any
time $t$, we have
$$
\man(t)=\lion_n(t)+\|\lion_n(t)\man(t)\|\cdot (\cos \xi(t),\sin \xi(t))
$$
and
$$
\goal=\lion_n(t)+\|\lion_n(t)\goal\|\cdot (\cos \eta(t),\sin \eta(t)).
$$

\begin{figure}
\centering
\includegraphics[scale=0.8]{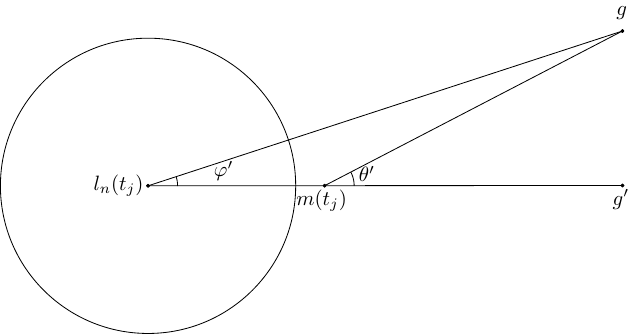}
\caption{The angles $\varphi'$ and $\theta'$ and the point $g'$ are shown.
The circle is $C(\lion_n(t_j),\rD)$.}
\label{phiThetaFig}
\end{figure}

Consider an arbitrary time of choice $t_j$ where
$\|\man(t_j)\lion_n(t_j)\|< \rD+\deltaT_n(1+\ee_n)$,
so that move $j$ is either an escape move or an avoidance move.
We now prove that a sufficient condition for move $j$ to be an escape move
is that the difference in angles
from $\lion_n$ to $\man$ and to $\goal$, respectively, is at most
$\varphi$. To put it in another way, if
\begin{align}\label{xieta}
\lvert \eta(t_j)-\xi(t_j)-2z\pi\rvert\leq \varphi
\end{align}
for some
$z\in\ZZ$, then the man makes an escape move. To see this, assume without loss of
generality that the inequality holds for $z=0$,
let $\varphi'\mydef\eta(t_j)-\xi(t_j)$, and consider the case
$0\leq\varphi'\leq\varphi$.
The case $0\geq\varphi'\geq-\varphi$ is analogous.
See Figure~\ref{phiThetaFig}.
Let $\theta'$ be the counterclockwise angle from the direction
$\overrightarrow{\lion_n(t_j)\man(t_j)}$ to $\overrightarrow{\man(t_j)\goal}$.
If
$\theta'\leq\theta=\arccos\frac 1{1+\ee_n}$, then move $j$ is an escape move:~The two lines
perpendicular to $\lion_n(t_j)\man(t_j)$ through
$\man(t_j)$ and $b$, respectively, are a witness, where $b$ is defined
as in case~\ref{choice:b}.
We now prove that since $\varphi'\leq\varphi$, we have
$\theta'\leq\theta$.
Let $g'$ be the projection of $g$ on the line through $\lion_n(t_j)$ and $\man(t_j)$.
We then have
\begin{align*}
\tan \theta'&=
\frac{\|gg'\|}{\|\man(t_j)g'\|}=
\frac{\|\lion_n(t_j)g\|\sin\varphi'}
{\|\lion_n(t_j)g\|\cos\varphi'-\|\lion_n(t_j)\man(t_j)\|}.
\end{align*}

Observe that under our assumptions,
$\theta'$ is maximum when $\varphi'$ and $\|\lion_n(t_j)\man(t_j)\|$
are maximum and
$\|\lion_n(t_j)g\|$ is minimum, i.e., when
$\varphi'=\varphi$,
$\|\lion_n(t_j)\man(t_j)\|=\rD+\deltaT_n$, and
$\|\lion_n(t_j)g\|=\rhoD$. We therefore get
\begin{align*}
\tan \theta'&\leq
\frac{\rhoD\sin\varphi}
{\rhoD\cos\varphi-(\rD+\deltaT_n)}\leq
\frac{\rhoD\sin\varphi}
{\rhoD\cos\varphi-2\rD}=\tan \theta.
\end{align*}
It now follows that $\theta'\leq\theta$, so indeed, move $j$ is an escape move.

In the following, we prove that there is some time of choice
$t'_j$ in the interval
$[t_i,t_i+\tau]$ for which the condition~\eqref{xieta} is satisfied, i.e.,
condition~\eqref{xieta} is true when $t_j$ is substituted by $t'_j$.

First, we note that for an arbitrary time of choice $t_j$ and any
$t\in[t_j,t_{j+1}]$, it holds that
\begin{align}\label{maxAngleStep}
\xi(t)-\xi(t_j)\leq
2\arcsin\frac{(2+\ee_n)\deltaT_n}{2(\rD-\deltaT_n)}.
\end{align}
To see this, we first observe that we can restrict ourselves to
the case $t=t_{j+1}$, since, clearly,
the maximum value of $\xi(t)-\xi(t_j)$ is attained for $t=t_{j+1}$.
Define for $t\in[t_j,t_{j+1}]$
$$
\lion'_n(t)\mydef\lion_n(t_j)\quad\text{and}\quad
\man'(t)\mydef\man(t)+(\lion_n(t)-\lion_n(t_j)),
$$
i.e., we fix the lion $\lion'_n$ at the point $\lion_n(t_j)$
and let the man $\man'$ run for both so that the segment
$\lion'_n(t)\man'(t)$ is a translation of $\lion_n(t)\man(t)$.
It follows that the man $\man'$ runs at speed at most $2+\ee_n$.
Clearly, $\xi$ cannot increase more on $[t_j,t_{j+1}]$ than in the case that
$\|\lion'_n(t_j)\man'(t_j)\|=\|\lion'_n(t_{j+1})\man'(t_{j+1})\|=\rD-\deltaT_n$
and $\|\man'(t_j)\man'(t_{j+1})\|=(2+\ee_n)\deltaT_n$. From this observation,
inequality~\eqref{maxAngleStep} follows from an elementary argument, see
Figure~\ref{xixiFig}.

\begin{figure}
\centering
\begin{minipage}[t]{.50\textwidth}
  \centering
  \includegraphics[scale=0.6]{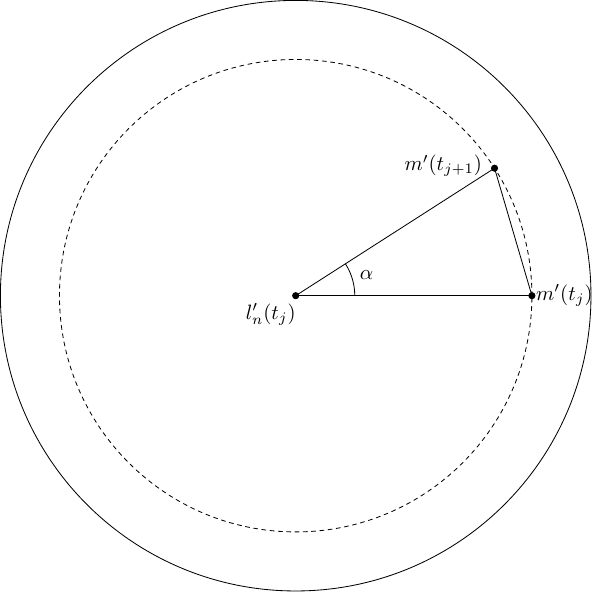}
  \captionof{figure}{The solid circle is $C(\lion'_n(t_j),\rD)$ and the dashed is
$C(\lion'_n(t_j),\rD-\deltaT_n)$.
The angle $\alpha$ is $\xi(t_{j+1})-\xi(t_j)$.
When
$\|\man'(t_j)\man'(t_{j+1})\|=(2+\ee_n)\deltaT_n$, then
$\alpha=2\arcsin\frac{(2+\ee_n)\deltaT_n}{2(\rD-\deltaT_n)}$.}
\label{xixiFig}
\end{minipage}\hspace{0.09\textwidth}%
\begin{minipage}[t]{.40\textwidth}
  \centering
  \includegraphics[scale=0.7]{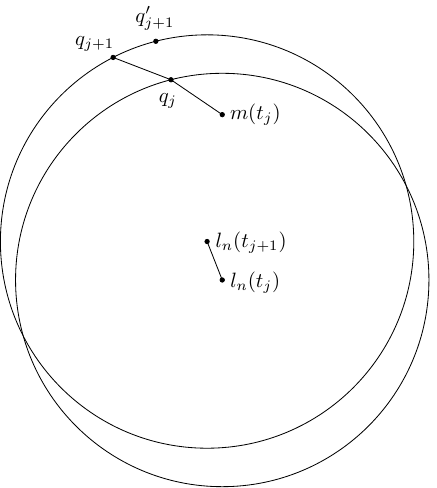}
\captionof{figure}{The circles are $C(\lion_n(t_j),\rD)$ and $C(\lion_n(t_{j+1}),\rD)$.
The man runs the path $\man(t_j)q_jq_{j+1}$. The angle from $q'_{j+1}$ to
$q_{j+1}$ on $C(\lion_n(t_{j+1}),\rD)$ is at least $\ee_n\deltaT_n$.}
\label{angleIncrease}
\end{minipage}\hspace{0.04\textwidth}
\end{figure}

We now note that for any $t\in[t_j,t_{j+1}]$, we have
\begin{align}\label{maxAngleStep2}
\eta(t)-\eta(t_j)\leq
\frac{\deltaT_n}{\rhoD}.
\end{align}
This inequality holds since $\eta$
is increasing the most when the lion runs counterclockwise
around the circle $C(g,\rhoD)$ with unit speed, in which case equality holds
in \eqref{maxAngleStep2} when $t=t_{j+1}$, and for smaller values of $t$,
the inequality is always strict.

Assume now that the moves $i,i+1,\ldots,i+k$ are all avoidance moves and
$t_{i+k}\leq t_i+\tau$.
See Figure~\ref{angleIncrease}.
For $j\in\{i,i+1,\ldots,i+k\}$, let
$q_j\mydef \man(t_{j+1})$ be the point to which the man
chooses to run at time $t_j$ as defined in case~\ref{choice:c}.
Let $\xi_j\in[\xi(t_j),\xi(t_j)+\pi]$ be the angle of $q_j$ on
$C(\lion_n(t_j),\rD)$, i.e., the angle such that
$$
q_j=\lion_n(t_j)+\rD\cdot (\cos \xi_j,\sin \xi_j).
$$
Let
$q'_{j+1}\mydef q_j+(\lion_n(t_{j+1})-\lion_n(t_j))$ be the point on the circle
$C(\lion_n(t_{j+1}),\rD)$ corresponding to $q_j$ on $C(\lion_n(t_{j}),\rD)$.
Then $\|\man(t_{j+1})q'_{j+1}\|=\|q_jq'_{j+1}\|=\|\lion_n(t_j)\lion_n(t_{j+1})\|
\leq\deltaT_n$. Hence,
$\|q_{j+1}q'_{j+1}\|\geq\|q_jq_{j+1}\|-\|q_jq'_{j+1}\|\geq (1+\ee_n)\deltaT_n-\deltaT_n
=\ee_n\deltaT_n$, and we get that $\xi_{j+1}-\xi_j>\frac{\ee_n\deltaT_n}{r}$ for any $j\in\{i,\ldots,i+k-1\}$ and hence that
$\xi_{i+k}-\xi_i> \frac{k\ee_n\deltaT_n}{r}$.

Clearly, $\xi_j$ is largest when $\man(t_j)$ is exactly on
the circle $C(\lion_n(t_j),\rD)$.
In this case,
the arc on $C(\lion_n(t_j),r)$ from $\man(t_j)$ to
$q_j$ is at most
$\frac{\pi \|\man(t_j)q_j\|}{2\rD}=\frac{\pi (1+\ee_n)\deltaT_n}{2\rD}$,
so we get that
$\xi_j\in\left[\xi(t_j),\xi(t_j)+
\frac{\pi(1+\ee_n)\deltaT_n}{2\rD}\right]$. Also note that inequality~\eqref{maxAngleStep2} gives
$\eta(t_{i+k})-\eta(t_i)\leq k\deltaT_n/\rho$.
Hence we have
\begin{align*}
D&\mydef\left(\xi(t_{i+k})-\xi(t_i)\right)-
\left(\eta(t_{i+k})-\eta(t_i)\right)\\
&\geq \left(\xi_{i+k}-\frac{\pi(1+\ee_n)\deltaT_n}{2\rD}-\xi_i\right)-
\frac{k\deltaT_n}{\rho}\\
&\geq \frac{k\ee_n\deltaT_n}{\rD}-\frac{\pi(1+\ee_n)\deltaT_n}{2\rD}-\frac{k\ee_n\deltaT_n}{2\rD}\\
&> \frac{k\ee_n\deltaT_n}{2\rD}-\pi.
\end{align*}
Now, if $k\geq\frac{6\pi\rD}{\ee_n\deltaT_n}$, we get
$D>2\pi$. Hence, after $\frac{6\pi\rD}{\ee_n\deltaT_n}\cdot\deltaT_n=\tau$ time,
$\xi$ has increased by at least $2\pi$ more than $\eta$. If follows that at some point
in time $t\in[t_i,t_i+\tau]$ and some $z\in\ZZ$, we have
$$
\lvert\xi(t)-\eta(t)-2z\pi\rvert = 0.
$$
Let $j\in\{i,\ldots,i+k\}$ be maximum such that $t_j\leq t$.
We now prove that Condition~\ref{xieta} is satisfied for the chosen $t_j$. Clearly,
$t-t_j\leq\deltaT_n$. It then follows from inequalities~\eqref{maxAngleStep} and
\eqref{maxAngleStep2} that
\begin{align*}
\lvert\xi(t_j)-\eta(t_j)-2z\pi\rvert &\leq
\lvert\xi(t_j)-\xi(t)\rvert+\lvert\xi(t)-\eta(t)-2z\pi\rvert+
\lvert\eta(t)+2z\pi-\eta(t_j)-2z\pi\rvert\\
&\leq
2\arcsin\frac{(2+\ee_n)\deltaT_n}{2(\rD-\deltaT_n)}+\frac{\deltaT_n}{\rho}\leq\varphi,
\end{align*}
and the claim has been proved.
\end{proof}

For $i\in\NN_0$, define the \emph{canonical interval} $I_i$
as $I_i\mydef\left[i\deltaT_{n-1},(i+1)\deltaT_{n-1}\right)$,
i.e., $I_i$ is the interval of time where the man would run from the $i$'th
to the $(i+1)$'st corner on the polygonal line defined by the strategy $\manStrat_{n-1}$.
We say that $I_i$ \emph{ends} at time
$t=(i+1)\deltaT_{n-1}$.
Note that if $t\in I_i$, then $\goal(t)=\manStrat_{n-1}((i+1)\deltaT_{n-1})$
and $\goal(t)$ moves when $I_i$ ends.

As a consequence of Claim~\ref{claim2} and Claim~\ref{claim3}, we get the following.

\begin{claim}\label{claim4}
If $t\in I_i$ and $\|\man(t)\goal(t)\|\leq\rhoD'$, then for
every $t'>t$, $t'\in I_i$, we have
$$\|\man(t')\goal(t)\|\leq \rhoD'+(1+\ee_n)\tau.$$
\end{claim}

\begin{claim}\label{claim5}
For any $i\in\NN_0$ and at any time during the canonical interval $I_i$,
the man is at distance at most $\rhoD'+2(1+\ee_n)\tau$ away from the segment
$\manStrat_{n-1}(I_i)$ and when $I_i$ ends, the man is within distance $\rhoD'+(1+\ee_n)\tau$ from
the endpoint $\manStrat_{n-1}((i+1)\deltaT_{n-1})$ of the segment.
\end{claim}

\begin{proof}
We prove the claim by induction on $i$.
To easily handle the base-case, we introduce
an auxiliary canonical interval $I_{-1}=[-\deltaT_{n-1},0)$ and assume that
the lions and the man are standing at their initial positions during all of
$I_{-1}$.
The statement clearly holds for $i=-1$.

Assume
inductively that the statement holds for $I_{i-1}$ and consider the interval $I_i$.
Let $\goal\mydef \manStrat_{n-1}((i+1)\deltaT_{n-1})$.
The additional distance that the man runs during $I_i$
when his speed is $1+\ee_{n}$ as compared to the speed
$1+\ee_{n-1}$ is $\deltaT_{n-1}(\ee_n-\ee_{n-1})$.
It follows from the definition of $\rD$ that
$$\deltaT_{n-1}(\ee_n-\ee_{n-1})\geq
\rhoD+2\rD+3(1+\ee_n)\tau>
\rhoD'+3(1+\ee_n)\tau.$$

By the induction hypothesis, the man is within a distance of
$\rhoD'+(1+\ee_n)\tau$ from $\manStrat_{n-1}(i\deltaT_{n-1})$
at time $i\deltaT_{n-1}$. Thus, his distance to $\goal$ at the beginning of
interval $I_i$ is at most
$\deltaT_{n-1}(1+\ee_{n-1})+\rhoD'+(1+\ee_n)\tau$,
where $\deltaT_{n-1}(1+\ee_{n-1})$ is the length of the interval
$M_{n-1}(I_i)$.
If the man does not do any avoidance moves during $I_i$, he runs
straight to $\goal$, so it follows that he reaches $\goal$ at time
$(i+1)\deltaT_{n-1}-2\tau$ at the latest.
Therefore, the statement is clearly true in this case.

Otherwise, let $t\in I_i$ be the first time of choice at which he does an avoidance move during $I_i$.
If he is at distance at most
$\rhoD'$ from $\goal$ at time $t$,
the statement follows from Claim~\ref{claim4}.
Therefore, assume that the distance is more than
$\rhoD'$.
Then, we must have that
$t<(i+1)\deltaT_{n-1}-2\tau$,
since, if $t$ was larger, he would already have reached $\goal$ by the above discussion.
Hence, Claim~\ref{claim3} gives that
at some time $t'\leq t+\tau$,
either
\begin{enumerate}
\item\label{claim5:1}
the man gets within a distance of $\rhoD'$
from $\goal$, or

\item\label{claim5:2}
he does an escape move.
\end{enumerate}

We first prove that in the interval $[t,t']$,
the distance from the man to the segment $\manStrat_{n-1}(I_i)$
is at most $\rhoD'+2(1+\ee_n)\tau$.
To this end, note that his distance to the segment at time $t$ is
at most
$\rhoD'+(1+\ee_n)\tau$. Thus, since $t'\leq t+\tau$, his distance at
time $t'$ can be at most $\rhoD'+2(1+\ee_n)\tau$.

It remains to be proven that the man stays within distance
$\rhoD'+2(1+\ee_n)\tau$ from $\manStrat_{n-1}(I_i)$ after
time $t'$ and that he is at distance at most
$\rhoD'+(1+\ee_n)\tau$ from $\goal$ at time $(i+1)\deltaT_{n-1}$.
If we are in case~\ref{claim5:1}, the statement follows from
Claim~\ref{claim4}, so assume case~\ref{claim5:2}.

By Claim~\ref{claim2}, the man will not do an avoidance move again after time $t'$ until
he reaches $\goal$ or $I_i$ ends.
While he is running directly towards $\goal$,
his distance to the segment $\manStrat_{n-1}(I_i)$ is decreasing,
so it follows that the distance is always at most
$\rhoD'+2(1+\ee_n)\tau$, as claimed.
Since he was doing avoidance moves in a period of length at most
$\tau$ before the escape move at time $t'$,
he can completely compensate for the
delay caused by the avoidance moves in the same amount of time by running
directly towards $\goal$.
The total delay is therefore at most
$2\tau$. Since he would reach $\goal$ at time
$(i+1)\deltaT_{n-1}-2\tau$ at the latest if he did not do any avoidance moves,
it follows that he reaches $\goal$
at time $(i+1)\deltaT_{n-1}$ or earlier.
The statement then follows from Claim~\ref{claim4}.
\end{proof}

We are now ready to finish our proof of Theorem~\ref{thm:fastMan}.
In particular, it remains to define safety distances $c_{nj}$ that satisfy the stated requirements, and to prove that the man eventually escapes the convex hull of the lions.
For any $j=1,\ldots,n-1$, we define $c_{nj}\mydef c_{n-1,j}-C_n$. 
Claim~\ref{claim5} says that during interval $I_i$ for any $i$, the distance from the man to the segment $\manStrat_{n-1}(I_i)$ is at most
\begin{align}\label{eq:boundingSegDist}
\rhoD'+2(1+\ee_n)\tau<
\rhoD+2\rD+2(1+\ee_n)\tau\leq C_n.
\end{align}

Consider a subinterval $I\mydef [p\deltaT_n,(p+1)\deltaT_n]$ of $I_i$ and the segment $\man(I)$ that the man traverses during $I$.
It follows from the induction hypothesis that the distance from any point on $\man(I)$ to lion $\lion_j$ for $j=1,\ldots,n-1$ is at least $c_{nj}$.
From the definition of $C_n$, we furthermore get that
\begin{align}\label{eq:boundingDist}
c_{nj}=
c_{jj}-\sum_{k=j+1}^n C_k\geq
c_{jj}\left(1-\sum_{k=1}^{n-j}1/2^{k+1}\right)\geq c_{jj}/2>0.
\end{align}

We define the safety distance to lion $\lion_n$ as $c_{nn}\mydef \rD-(3+\ee_n)\deltaT_n$.
By Claim~\ref{claim1}, it holds that in the time interval $I$, the distance from any point on the segment $\man(I)$ to $\lion_n$ is at least $c_{nn}$.

To see that the man escapes the convex hull of the lions, we first give a bound on the distance $\|M_{n-1}(t)M_n(t)\|$ for $t\in I_i$.
Since $M_{n-1}(t)$ is a point on the segment $\manStrat_{n-1}(I_i)$ of length $\deltaT_{n-1}(1+\ee_n)$, the bound~\eqref{eq:boundingSegDist} on the distance from $M_n(t)$ to the segment $\manStrat_{n-1}(I_i)$ implies that
\begin{align}\label{eq:manmanDist}
\|M_{n-1}(t)M_n(t)\|\leq C_n+\deltaT_{n-1}(1+\ee_n)<2C_n.
\end{align}
It follows that $\|\manStrat_1(t)\manStrat_n(t)\|\leq 2\sum_{k=2}^n C_k$.
Since $\manStrat_1(t)$ traverses a ray with constant speed
$1+\ee/2>1$, the man eventually escapes the convex hull and the distance to the convex hull diverges to $\infty$ as $t\longrightarrow\infty$.
This proves the theorem.
\end{proof}

\section{The Fast Man Surviving Infinitely many Lions}\label{infLionsSec}

We now show that taking the limit of the strategy $M_n$ from the proof of Theorem~\ref{thm:fastMan} as $n\longrightarrow\infty$ gives a strategy $M_\infty$ for the man winning over any countably infinite set of lions.

\begin{theorem}
    In the plane $\mathbb{R}^2$, for any $\ee>0$, a man able to run at speed $1+\ee$ has a winning strategy
    against any countably infinite set $\{\lion_1,\lion_2,\ldots\}$ of unit-speed lions, provided that the man does not start at the same point as a lion.
\end{theorem}

\begin{proof}
We consider the strategy $M_n$ from the proof of Theorem~\ref{thm:fastMan} as $n\longrightarrow\infty$.
We first prove that at any point in time $t$, the point $M_n(t)$ converges to a point $M_\infty(t)$ as $n\longrightarrow\infty$.
For $m>n$, inequality~\eqref{eq:manmanDist} yields that
$$
\|M_n(t)M_m(t)\|\leq \sum_{i=n+1}^m 2C_i\leq 2c_{11} \sum_{i=n+1}^m\frac 1{2^i}.
$$
The latter sum tends to $0$ as $n,m\longrightarrow\infty$.
Hence, each point $M_\infty(t)$ is well-defined.
We need to ensure that $M_\infty(t)$ moves with speed at most $1+\ee$.
This is indeed the case since for any two points in time $s,t$, we have for any $n$ that $\|M_n(s)M_n(t)\|\leq(1+\ee)\cdot|s-t|$.
Therefore, it must also be the case that $\|M_\infty(s)M_\infty(t)\|\leq(1+\ee)\cdot|s-t|$.
At last, we need to check that the distance from the man to each lion $\lion_j$ is bounded by some constant, so that $\lion_j$ never catches the man.
This is exactly the statement in inequality~\eqref{eq:boundingDist} as $n\longrightarrow\infty$.
\end{proof}

\section{Concluding Remarks}

We conclude the paper by suggesting some open problems.
The proof that three lions are enough to catch the man relies on the assumption that there are only finitely many lakes.
If there are infinitely many lakes, a similar strategy might not work, as the lions may never restrict the man to a region with no lakes.
It is an interesting question if three lions are also enough if there are infinitely many lakes.
It must be assumed that the sum of the lengths of all the boundaries of the lakes is finite, as otherwise, the initial distance from the lions to the man can be infinite, even though each lake is rectifiable.

We have shown that a fast man can survive any finite number of lions, and that the lions do not also have a winning strategy since the man's strategy is locally finite.
Furthermore, the man has a winning strategy against any countably infinite set of lions.
This strategy, however, is not locally finite, so the lions might also have a winning strategy, but we have not been able to find one.
For the man to have a winning strategy, it is necessary that there are ``only'' countably many lions.
Indeed, the man cannot win if the lions form a complete circle centered at the man.





\bibliographystyle{plain}
\bibliography{bib}

\begin{thebibliography}{10}

\bibitem{abrahamsen2017best}
Mikkel Abrahamsen, Jacob Holm, Eva Rotenberg, and Christian Wulff-Nilsen.
\newblock {Best Laid Plans of Lions and Men}.
\newblock In {\em 33rd International Symposium on Computational Geometry (SoCG
  2017)}, pages 6:1--6:16, 2017.

\bibitem{abrahamsen2020escaping}
Mikkel Abrahamsen, Jacob Holm, Eva Rotenberg, and Christian Wulff-Nilsen.
\newblock Escaping an infinitude of lions.
\newblock {\em The American Mathematical Monthly}, 127(10):880--896, 2020.

\bibitem{aigner1984game}
Martin Aigner and Michael Fromme.
\newblock A game of cops and robbers.
\newblock {\em Discrete Applied Mathematics}, 8(1):1--12, 1984.

\bibitem{alon2015chasing}
Noga Alon and Abbas Mehrabian.
\newblock Chasing a fast robber on planar graphs and random graphs.
\newblock {\em Journal of Graph Theory}, 78(2):81--96, 2015.

\bibitem{berarducci1993cop}
Alessandro Berarducci and Benedetto Intrigila.
\newblock On the cop number of a graph.
\newblock {\em Advances in Applied Mathematics}, 14(4):389--403, 1993.

\bibitem{bhadauria2012capturing}
Deepak Bhadauria, Kyle Klein, Volkan Isler, and Subhash Suri.
\newblock Capturing an evader in polygonal environments with obstacles: The
  full visibility case.
\newblock {\em The International Journal of Robotics Research},
  31(10):1176--1189, 2012.

\bibitem{bollobas2012lion}
B.~Bollob{\'a}s, I.~Leader, and M.~Walters.
\newblock Lion and man---can both win?
\newblock {\em Israel Journal of Mathematics}, 189(1):267--286, 2012.
\newblock Version with appendix: \url{https://arxiv.org/abs/0909.2524}.

\bibitem{bollobas2006art}
B{\'e}la Bollob{\'a}s.
\newblock {\em The Art of Mathematics: Coffee Time in Memphis}.
\newblock Cambridge University Press, 2006.

\bibitem{bollobas2011lion}
B{\'e}la Bollob{\'a}s.
\newblock The lion and the christian, and other pursuit and evasion games.
\newblock In Dierk Schleicher and Malte Lackmann, editors, {\em An Invitation
  to Mathematics: From Competitions to Research}, pages 181--193.
  Springer-Verlag Berlin Heidelberg, 2011.

\bibitem{CHERNOUSKO1976problem}
F.L. Chernous'ko.
\newblock A problem of evasion from many pursuers.
\newblock {\em Journal of Applied Mathematics and Mechanics}, 40(1):11--20,
  1976.

\bibitem{chung2011search}
Timothy~H. Chung, Geoffrey~A. Hollinger, and Volkan Isler.
\newblock Search and pursuit-evasion in mobile robotics.
\newblock {\em Autonomous Robots}, 31:299--316, 2011.

\bibitem{croft1964lion}
Hallard~T. Croft.
\newblock ``{L}ion and man'': A postscript.
\newblock {\em Journal of the London Mathematical Society}, 39:385--390, 1964.

\bibitem{flynn1973lion}
James Flynn.
\newblock Lion and man: The boundary constraint.
\newblock {\em SIAM Journal on Control}, 11:397--411, 1973.

\bibitem{flynn1974lion}
James Flynn.
\newblock Lion and man: The general case.
\newblock {\em SIAM Journal on Control}, 12:581--597, 1974.

\bibitem{fokkink2013open}
Robbert Fokkink, Leonhard Geupel, and Kensaku Kikuta.
\newblock Open problems on search games.
\newblock In Steve Alpern, Robbert Fokkink, Leszek~Antoni G\k{a}sieniec, Roy
  Lindelauf, and V.S. Subrahmanian, editors, {\em Search Theory: A Game
  Theoretic Perspective}, chapter~5, pages 181--193. Springer-Verlag New York,
  2013.

\bibitem{fomin2010pursuing}
Fedor~V. Fomin, Petr~A. Golovach, Jan Kratochv\'{i}l, Nicolas Nisse, and Karol
  Suchan.
\newblock Pursuing a fast robber on a graph.
\newblock {\em Theoretical Computer Science}, 411:1167--1181, 2010.

\bibitem{ibragimov2012evasion}
Gafurjan~I. Ibragimov, Mehdi Salimi, and Massoud Amini.
\newblock Evasion from many pursuers in simple motion differential game with
  integral constraints.
\newblock {\em European Journal of Operational Research}, 218(2):505--511,
  2012.

\bibitem{jankovic1978about}
Vladimir Jankovi{\'c}.
\newblock About a man and lions.
\newblock {\em Matemati\v{c}ki Vesnik}, 2:359--361, 1978.

\bibitem{lewin1986lion}
J.~Lewin.
\newblock The lion and man problem revisited.
\newblock {\em Journal of Optimization Theory and Applications},
  49(3):411--430, 1986.

\bibitem{littlewood1986littlewood}
John~Edensor Littlewood.
\newblock {\em Littlewood's miscellany: edited by B{\'e}la Bollob{\'a}s}.
\newblock Cambridge University Press, 1986.

\bibitem{rado1974more}
Peter~A. Rado and Richard Rado.
\newblock More about lions and other animals.
\newblock {\em Mathematical Sprectrum}, 7(3):89--93, 1974/75.

\end{thebibliography}


\end{document}